\documentclass[10pt]{article}

\usepackage{amsmath,amsthm,verbatim,amssymb,amsfonts,amscd, graphicx, enumitem}
\usepackage{graphics}
\usepackage{centernot}
\usepackage{authblk}
\theoremstyle{plain}
\newtheorem{theorem}{Theorem}

\newtheorem{lemma}{Lemma}
\newtheorem*{remark}{Remark}

\newtheorem{definition}{Definition}
\newtheorem{assumption}{Assumption}

\usepackage[colorlinks,linkcolor=blue,citecolor=blue]{hyperref}

\usepackage[bottom]{footmisc}
\usepackage{caption}
\usepackage{subcaption}
\usepackage{helvet}  
\usepackage{courier}  
\usepackage{url}  
\usepackage{graphicx}  
\usepackage{multirow}
\usepackage{amsthm}
\usepackage{color}
\usepackage{MnSymbol}
\usepackage{makecell}
\usepackage{arydshln}
\usepackage{amsmath}
\usepackage[dvipsnames]{xcolor}
\usepackage{caption} 
\usepackage{natbib}
\usepackage{bbm}

\usepackage{textcomp}
\usepackage{wrapfig}
\usepackage{algorithm}
\usepackage{algorithmic}

\usepackage{csquotes}

\begin{document}

\title{Heterosynaptic Circuits Are Universal Gradient Machines}

\author{Liu Ziyin$^{1,2}$, Isaac Chuang$^1$, Tomaso Poggio$^1$\\
$^1$\textit{Massachusetts Institute of Technology}\\
$^2$\textit{NTT Research}
}

\maketitle

\begin{abstract}
We propose a design principle for the learning circuits of the biological brain. The principle states that almost any dendritic weights updated via heterosynaptic plasticity can implement a generalized and efficient class of gradient-based meta-learning. The theory suggests that a broad class of biologically plausible learning algorithms, together with the standard machine learning optimizers, can be grounded in heterosynaptic circuit motifs. This principle suggests that the phenomenology of (anti-) Hebbian (HBP) and heterosynaptic plasticity (HSP) may emerge from the same underlying dynamics, thus providing a unifying explanation. It also suggests an alternative perspective of neuroplasticity, where HSP is promoted to the primary learning and memory mechanism, and HBP is an emergent byproduct. We present simulations that show that (a) HSP can explain the metaplasticity of neurons, (b) HSP can explain the flexibility of the biology circuits, and (c) gradient learning can arise quickly from simple evolutionary dynamics that do not compute any explicit gradient. While our primary focus is on biology, the principle also implies a new approach to designing AI training algorithms and physically learnable AI hardware. Conceptually, our result demonstrates that contrary to the common belief, gradient computation may be extremely easy and common in nature.
\end{abstract}

\section{Introduction}

What kind of learning algorithm might the brain actually implement? To be biologically plausible, a learning rule must satisfy three key requirements. First, it needs to be local: a neuron cannot access information beyond its neighborhood -- which consists of its neighboring neuron activations and local densities of ions and neuromodulators. 
Second, it must be flexible: the algorithm should not rely on finely tuned or rigid connectivity patterns but instead should operate effectively across a wide range of network topologies -- from highly structured and seemingly random topologies \cite{anderson2022big}. Third, it should be robust: learning should be reliable despite both microscopic imperfection in the neurons and macroscopic changes in the architecture -- a property often empirically observed in the biological brain, where learning remains possible even if either the firings are low-precision or some brain regions are damaged \cite{fertonani2011random, giovagnoli1999learning}. These broad requirements raise a fundamental question: can there be a universal design principle that satisfies all locality, flexibility, and robustness conditions while supporting efficient learning comparable to what is achieved in artificial neural networks?

In addition, many biological observations regarding neuroplasticity constrain how learning rules must operate in the brain \cite{chistiakova2014heterosynaptic, jungenitz2018structural, turrigiano2000hebb, andersen2017hebbian}. Four seemingly unrelated types of plasticity exist. First, Hebbian and anti-Hebbian plasticity states that synaptic weights are updated homosynaptically to reinforce emergent correlations. Second, heterosynaptic plasticity  \cite{chistiakova2014heterosynaptic} exists to allow synapses not directly involved in a particular signal pathway to be modulated by other neurons. Third, plasticity of neurons depend on its prior experience and stimulus, a phenomenon known as metaplasticity \cite{abraham2008metaplasticity}. Fourth, 
on top on these observations, 
gradient descent in artificial neural networks have been found to give rise to emergent representations quite comparable to biological brains \cite{yamins2016using}. 
These observations of plasticity motivate a deeper question: can we identify a single learning rule that unifies these seemingly unrelated plasticity rules -- one that appears Hebbian, Heterosynaptic, meta-learning, and gradient-like at the same time?

Our result shows that generic two-signal update rules, which can be seen as the a generalized form of heterosynaptic plasticity, are sufficient to simulate gradient learning as long as some kind of homeostatic property is achievable. Our theory shows that when the synapse for one signal becomes stable (heterosynaptic stability), it gives rise to a local emergent scalar effective ``learning rate" we call the ``consistency score" for the other signal. The local consistency score determines the sign of learning and will interact with its neighbors in a way that it will become consistent globally consistent (dynamical consistency). This theory implies a design principle that utilizes heterosynaptic plasticity to enable a spectrum of learning algorithms--from highly specialized, topologically structured modules to flexible, randomly connected networks -- to simulate gradient descent universally. See Figure~\ref{fig:neuro illustration} for examples of such circuits in the brain and how our proposed mechanism enables it to learn. 
Formal theorems, proofs, and additional figures are in the Appendix.

\begin{figure}[t!]
    \centering
    \includegraphics[width=0.8\linewidth]{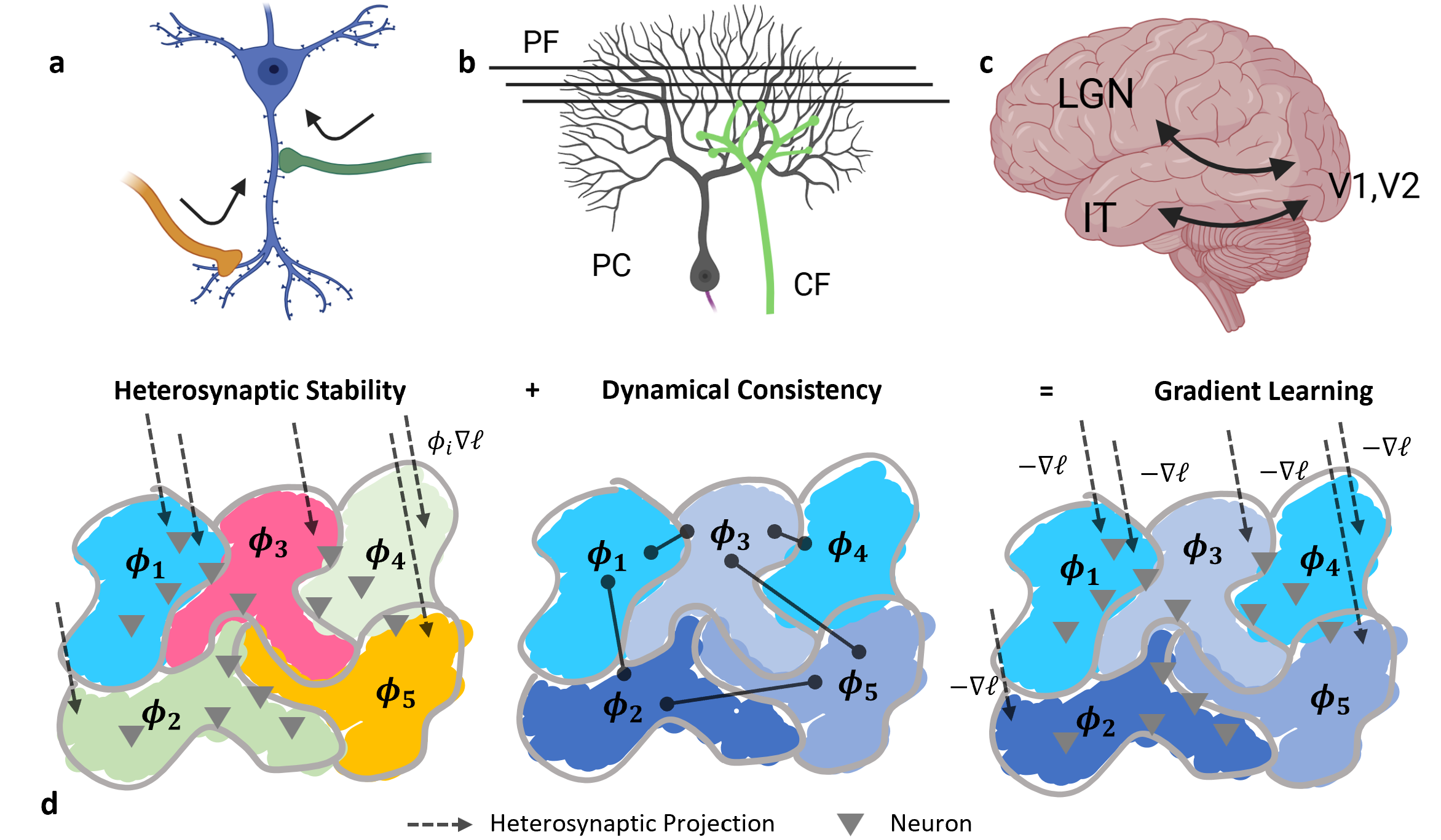}
    \caption{\small Microscopic and macroscopic structures of biological heterosynaptic circuits (\textbf{a}-\textbf{c}) could implement gradient learning with the proposed ``HSDC" mechanism (\textbf{d}). \textbf{a}: The minimal structure required to build a heterosynaptic circuit is a neuron with two incoming signals. Note that it does not require two inputs -- it could be a single axon that fires twice at different times. Due to its simple compositional nature, this circuit can be biologically realized across vastly different scales. \textbf{b}: Microscopically, the Purjinke cells (PC) in the mammalian cerebellum has a heterosynaptic motif, where parallel fiber (PF) carries the sensory input while climbing fiber (CF) carries learning or error signals. \textbf{c}: At a large scale, the bi-pathway structures between cortical regions can also implement heterosynaptic circuits. \textbf{d}: Heterosynaptic stability (HS) and dynamical consistency (DC) are sufficient to enable gradient learning. HS gives rise to local patches of neurons characterized by a scalar consistency score $\phi$; DC guarantees that $\phi$ has the same sign for all patches. For example, in \textbf{c}, the CF functions as heterosynaptic projections, which could give rise to $\phi$ values for the PC. In \textbf{d}, the feedback pathway from V1 to LGN can function as a heterosynaptic projection, which could give rise to consistency scores among the feedforward neurons or cortical columns in V1.}
    \label{fig:neuro illustration}
\end{figure}

\begin{figure}[t!]
    \centering
    \includegraphics[width=\linewidth]{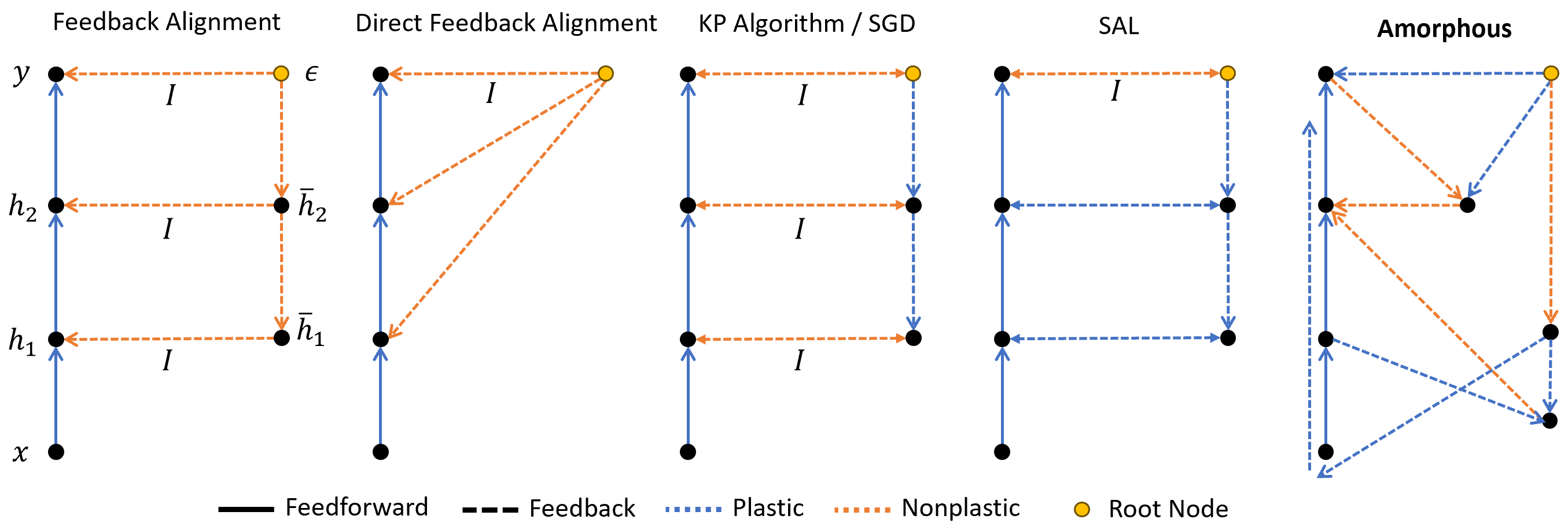}    \caption{\small Examples of heterosynaptic circuits for training a two-hidden-layer neural network \cite{nokland2016direct, lillicrap2016random, akrout2019deep, kolen1994backpropagation}. One can imagine a heterosynaptic circuit as a superposition of two graphs, one performing computation at time $t$ (solid) and the other performing computation at time $t'$ (dashed). In the figure, every node is a set of neurons, and every edge corresponds to a dense matrix sending connecting two such nodes. SGD can be seen as the special case of the KP algorithm when the two pathways are initialized symmetrically. Each edge can be either plastic or nonplastic. When nonplastic, the edges are fixed to be identity or random matrices. The rightmost (\textbf{amorphous}) shows an example of the most general type of heterosynaptic circuit to train such a network. Our theory shows that all the above circuits can simulate gradient learning.}
    \label{fig:circuit examples}
\end{figure}

\section{Heterosynaptic Circuits are Universal Gradient Machines}\label{sec: main result}

Suppose we have an environment that provides the input-label pair $(x,y)$, where $y$ could be identical to $x$ for unsupervised learning. Our hypothetical learning brain will contain a learning circuit $f$ (Figure~\ref{fig:neuro illustration})  that computes a prediction $\hat{y} = f(x)$; to achieve this, $f$ necessarily consists of three types of neurons: (1) input and output neurons, which interact with the environment; (2) a set of latent neurons that are mutually connected; (3) a root node (perhaps identical to the output neurons) that receives an initial error signal $\epsilon=\epsilon(y,\hat{y})$. 

Though unnecessary in theory, the latent neurons could be further divided into two types: (a) forward neurons $h$ that are functions of $x$ and directly affect $f(x)$: $f(x) = f(h(x))$, and (b) auxiliary neurons $\bar{h}$ that are used for enabling learning, which are functions of both $\epsilon$ and $x$. We will also refer to $\bar{h}$ as ``backward neurons" because they are sometimes imagined to constitute the backward pathway in cortices. Note that $h$ and $\bar{h}$ may share neurons. Recent works conjecture that learning signals can travel through separate feedback pathways to train the predictive model \cite{lillicrap2020backpropagation, whittington2019theories}. This separation into bidirectional pathways is a common anatomy in the brain, which may be realized as high-level structures, such as the feedforward-feedback pathways between V1 and V2 \cite{markov2014anatomy} or as composable microcircuits. Such a microcircuit exists, e.g., in the cerebellum (Figure~\ref{fig:neuro illustration}), where the dendritic weights are updated by a heterosynaptic circuit \cite{magee2020synaptic}. 

In this work, we use the letter $p$ to denote the preactivation of neurons and $h$ to denote the postactivation of neurons such that $h(p)$ is a nonlinear function. We mainly focus on ReLU-type activations in the main text and discuss how the results are generalizable to general activations in the appendix. 
Now, let $p$ be the preactivation of a subset of neurons that are densely connected to another subset of neurons with postactivations $\tilde{h}$. At time $t$, $p$ receives input from $\tilde{h}$ and is given by
\begin{equation}
    p(x,t) = W \tilde{h}(x,t),
\end{equation}
where $W$ is a weight matrix connecting $p$ and $\tilde{h}$, and $x$ is the input signal. Because $p$ are always functions of $x$, we omit $x$ when the context is clear. 
Similarly, at a different time $t'$, the neuron receives another signal from a set of potentially different neurons $p(t') =  \bar{V}  \bar{h}(t')$. Note that these equations can be seen as stationary points of an underlying differential equation: $\dot{p} = \bar{V} \bar{h} - p$. One can imagine this process as a superposition of two computation graphs on the same set of nodes, where the first graph computes at time $t$ and the second graph computes at time $t'$. See Figure~\ref{fig:meta-and-micro}-a1 for an illustration. 

Now, the synaptic weights $\bar{V}$ and $W$ are updated according to the simplest type of heterosynaptic rule:
\begin{equation}\label{eq: update rule}
\begin{cases}
    \Delta \bar{V} = \eta p(t) \bar{h}^\top(t')  - \gamma \bar{V},\\
    \Delta W = \eta p(t') \bar{h}^\top(t)  - \gamma W,
\end{cases}
\end{equation}
where $\eta$ is the learning rate, and $\gamma$ is a weight decay that prunes synapses when there is no firing. That two different incoming signals could play different roles is supported by experimental results, where the dendrites of a single neuron can be divided into multiple compartments (e.g., apical vs. basal) and deal with signals in different ways \cite{harnett2013potassium}. Note a key feature is that the update for either matrix is the product of signals from both time $t$ and $t'$. Most existing biological learning algorithms (including backpropagation) are either explicitly or implicitly a heterosynaptic rule (Figure~\ref{fig:circuit examples}). Because heterosynaptic rules always involves a neuron receiving at least two input signals, we use ``heterosynaptic rule" interchangeably with ``two-signal rule" from now on. If we model the dynamics of $p$ explicitly, one can replace the term $p(t)$ in $\Delta \bar{V}$ to be $\dot{p}(t)$, and so the rule is also consistent with temporal error models of plasticity \cite{whittington2019theories}.


Our theory below shows that this class of two-signal processes in \eqref{eq: update rule} is sufficient to enable global gradient learning when two local properties are met: (a) Heterosynaptic Stability (HS) and (b) Dynamical Consistency (DC). The central thesis of this work is the following ``HSDC" formula:
\begin{quote}\centering
     HS + DC = Gradient Learning
\end{quote}
HS is essentially a form of homeostasis and a property of the update rule and its dynamics. DC is a property of the circuit architecture, or of the learning rule, or could be emergent. We will show that HS + DC give rise to the following set of equations:
\begin{align}
        &{\rm HS}:\quad p_i(t')=\phi_i H \nabla_{p_i(t)} \ell(p_i(t)), \label{eq: hs}\\
       &{\rm DC}:\quad \phi_i \phi_j \geq 0, \quad \forall i,\ j, \label{eq: dc}
\end{align}
where $i,\ j$ are the indices of the nodes and $H$ is a matrix learning rate, which is common in machine learning (see Appendix~\ref{app sec: indistinguishability}). $H$ is plastic and so constitutes a form of meta-learning. $\phi_i\in \mathbb{R}$ is the consistency score. While $p(t')$ is a gradient for the activation, it leads to a computation of $\nabla_W \ell$ in Eq.~\eqref{eq: update rule} by the chain rule (Appendix~\ref{app sec: neuron gradient descent}). The first equation shows that HS leads to learning with a matrix learning rate $H$ and modulated by a consistency score. The second equation ensures that different nodes are have the same learning direction.

In the rest of the paper, We show numerical results as we present the theory. We first argue that the HSP motifs is ubiquitous in learning algorithms and so they can be easily regarded as minimal functioning elements of these algorithms. We then discuss how and why HS and DC lead to Eq.~\eqref{eq: hs} and \eqref{eq: dc}. Because the theory implies infinitely many ways to construct a learning circuit that simulates gradient descent, it is impossible to sample all circuits to test our theory. Because the the SAL architecture (Figure~\ref{fig:circuit examples}) resembles both SGD and feedback alignment, we regard it as the ``canonical" architecture and sample random topologies that are either microscopic or macroscopic variations of it. See Appendix~\ref{app sec: exp} for details. 


\subsection{Two-Signal Circuits are Ubiquitous}
The easiest way to construct a heterosynaptic circuit is using a ``bi-pathway" architecture. In fact, SGD \textit{and} most of the feedback-type algorithms can be thought of as a special case of the algorithm in Eq.~\eqref{eq: update rule}. It is known that SGD (or, backpropagation) is equivalent to a circuit having a forward pathway and a backward pathway \cite{akrout2019deep}, where the forward pathway is the network for inference, while the backward pathway contains the transpose of the forward synaptic weights and computes the gradient of each layer sequentially in the reverse direction. Here, the update rule is the outer product of the activation from the forward pathway and the activation  of the backward pathway. Thus, the update rule to $W$ is a heterosynaptic rule, and for those $x$ that this update rule is stationary, the update rule is a gradient learning rule -- this is certainly consistent with the fact that SGD is a gradient learning rule. See Figure~\ref{fig:circuit examples}.

Feedback-type learning rules are also special cases of this algorithm. The standard feedback alignment \cite{lillicrap2016random} obeys the following dynamics
$\Delta W = \bar{h}h^\top -\gamma W$. Identifying $W$ with $\bar{V}^\top$, we have that $\bar{h}\propto W^\top\nabla_{h} \ell$, which is the gradient for $\nabla_{Wh} \ell$. The same argument applies to directed feedback, where all gradient signals come from the root node. Similarly, the argument applies to the KP algorithm \cite{kolen1994backpropagation} and the recent SAL algorithm \cite{liao2024self}, where the four-way motif can be seen as a composition of four simple heterosynaptic motifs. As shown in Ref.~\cite{liao2024self}, it is advantageous to train both the forward and backward pathways, so the interconnections not only go from the backward to the forward, but also from the forward to the backward. Ref.~\cite{liao2024self} also showed that if the feedback pathway is much wider than the feedforward, the model performance will be improved, partially explaining the anatomical observation of a more numerous feedback connection \cite{briggs2020role}.

Because Eq.~\eqref{eq: hs} does not depend on the computation of $\bar{h}$, the theory allows an essentially arbitrary backward pathway. One can replace the backward pathway with any network with an arbitrary connectivity and activation function and gradient learning is still possible. For example, one can use a transformer, RNN, or CNN, or a mixture of these. One can also reuse a part of the forward network. For example, we can connect the root node directly to the input node of the forward pathway, and reuse the forward network to pass on the signal. This works because the update is not ``simultaneously" homosynaptic. To see why this works, note that during learning, the signal received by the input node is $\bar{V}\bar{V}^\top \nabla_x \ell$. After passing through the first layer, the received gradient is $W\bar{V}\bar{V}^\top \nabla_x \ell = W\bar{V}\bar{V}^\top  W^\top\nabla_{Wx} \ell$, which is the gradient of the first layer with a PSD matrix learning rate. This argument can be applied recursively to show that it works for every layer. 
Thus, the most general circuit one can construct is perhaps an amorphous model with an arbitrary graph structure and may also reuse part of the forward pathway during computation. 

\begin{figure}[t!]
    \centering
    \includegraphics[width=1.0\linewidth]{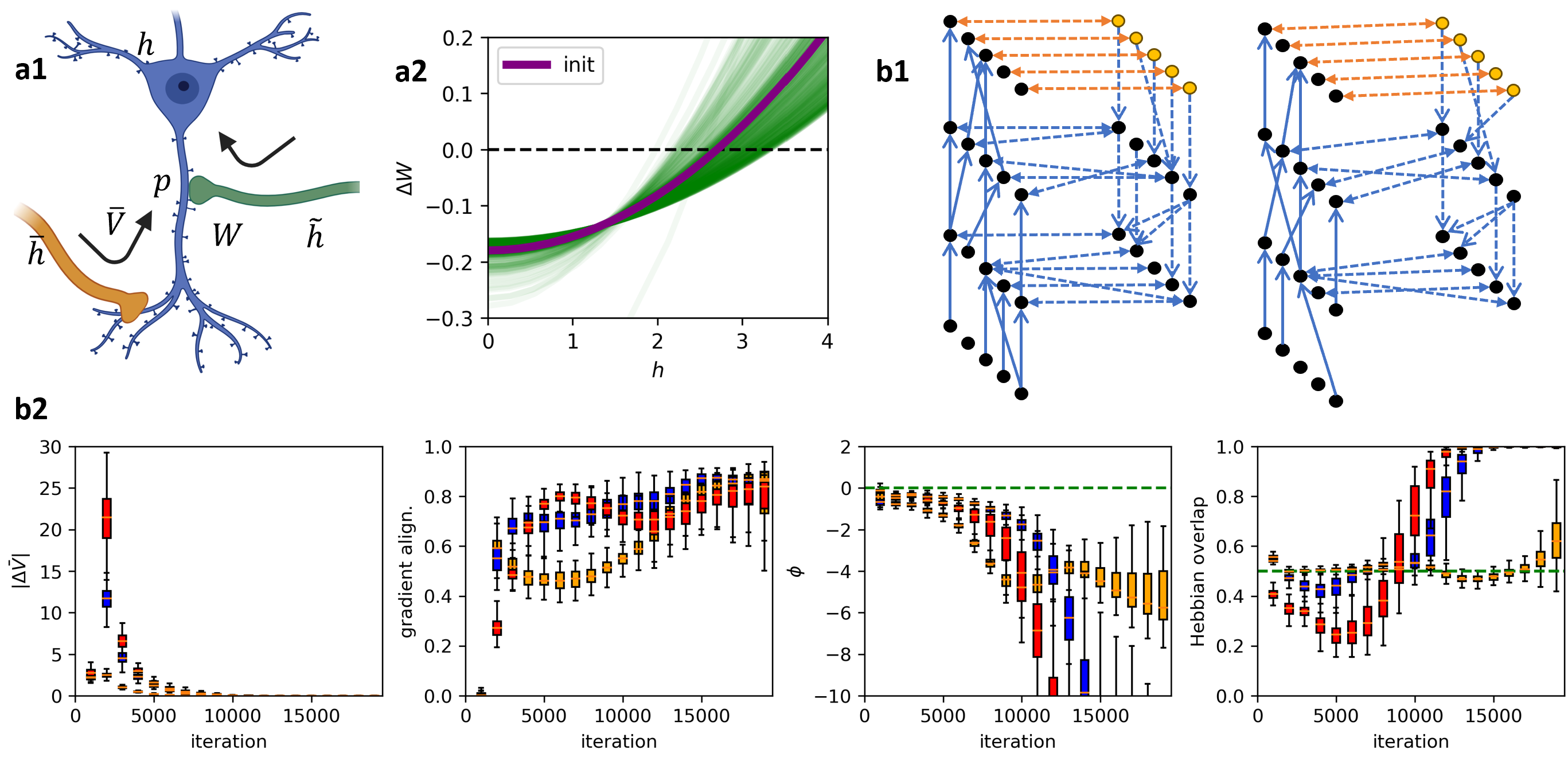}
    \caption{\small Meta-plasticity of neurons (\textbf{a}) and emergence of gradient learning in networks with microscopic random connectivity (\textbf{b}). \textbf{a1}: A neuron learning simple task with a root node $\tilde{h}$ and input node $\bar{h}$. \textbf{a2}: Plasticity of a synapse after it sees two samples (green curve shows 400 runs). The purple  curve shows the plasticity of the synapse before it seem any data point. Prior experience alters the plasticity threshold for the synapse, consistent with biological observations \cite{abraham2008metaplasticity}. \textbf{b1}: We train a four-hidden layer neuron network (two shown in the figure), where only connections between neighboring layers are allowed. Within the allowed connections, we randomly generete masks with $70\%$ density to simulate a different learning circuit topology. The figure illustrates two instances of such sampled connectivities. The color and style have the same meaning as in Figure~\ref{fig:circuit examples}. \textbf{b2}: the learning trajectories across $2\times 10^4$ steps for $100$ i.i.d. randomly sampled learning circuits. Colors correspond to different layers: orange (1), blue (3), red (4). Gradient alignment is the correlation of the HSP update with the actual negative gradient. Hebbian overlap is the fraction of updates that have a positive alignment with the Hebbian update See Figure~\ref{fig:V stationarity and alignment}-\ref{fig:hebbian overlap zoom in} for more detailed results.}
    \label{fig:meta-and-micro}
\end{figure}

\begin{figure}[t!]
    \centering
    \includegraphics[width=0.25\linewidth]{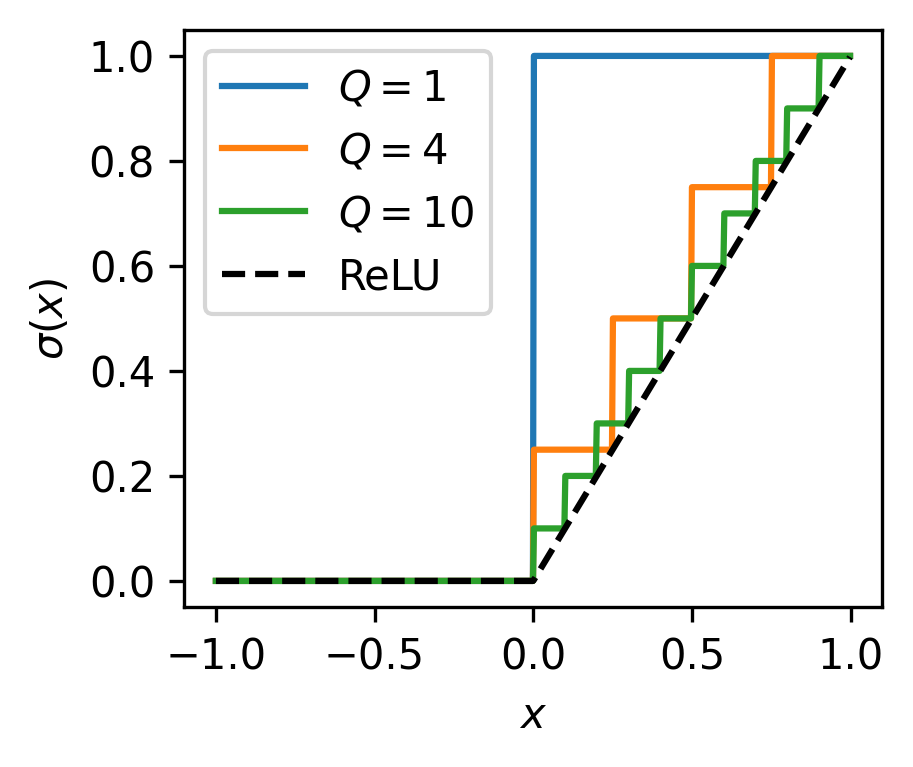}
    \includegraphics[width=0.25\linewidth]{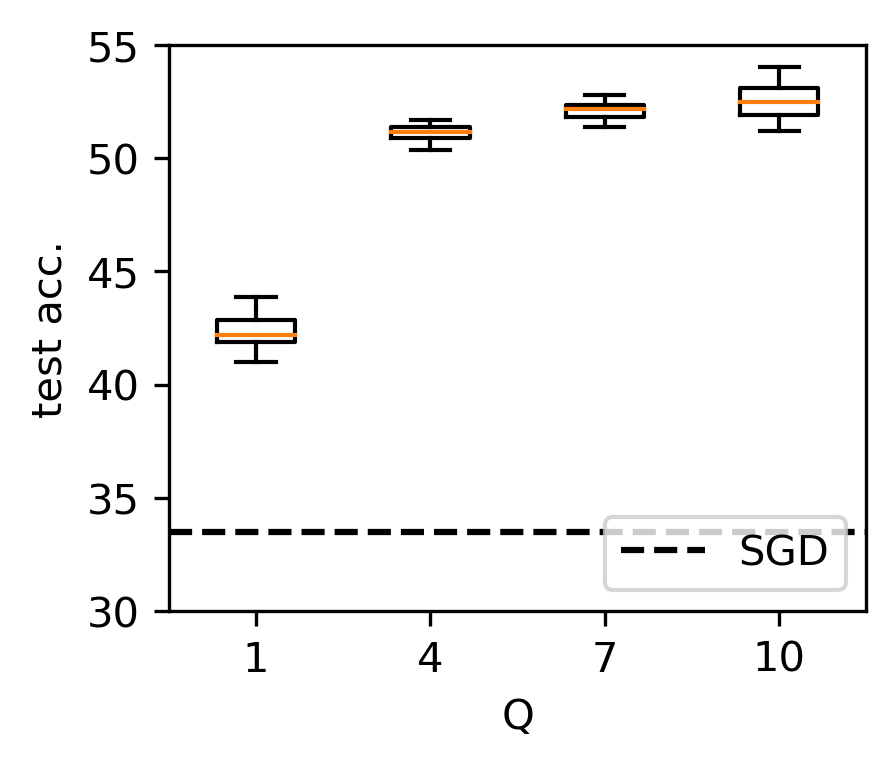}
    \includegraphics[width=0.25\linewidth]{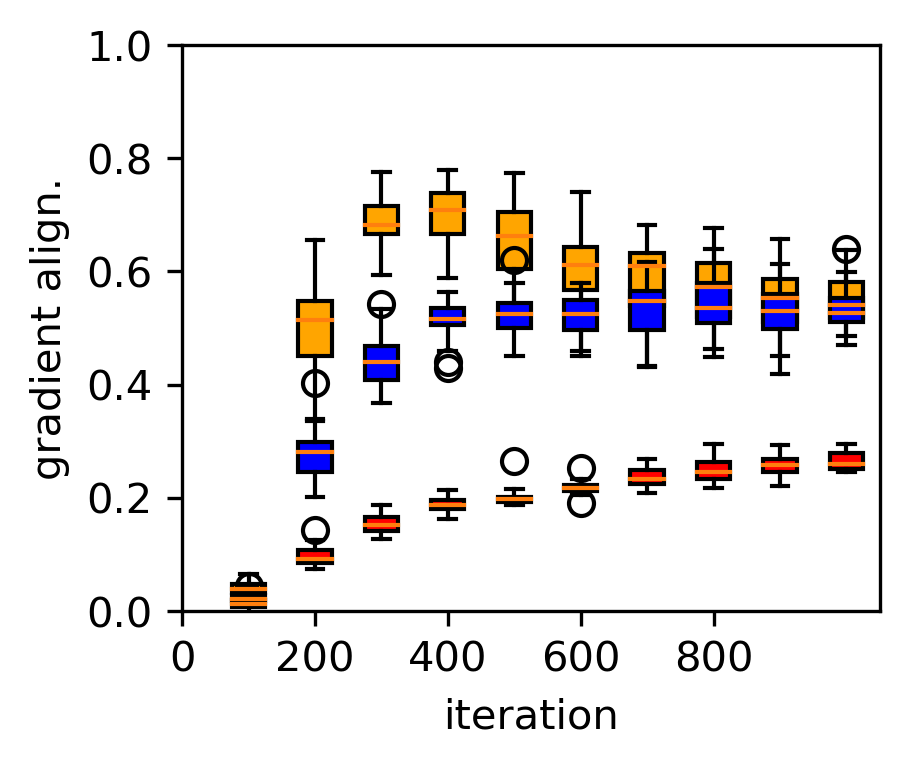}

    \caption{\small Training of a deep network with the nondifferentiable step-ReLU activation. \textbf{Left}: Illustration of the step-ReLU activation: $\sigma(x) = {\rm ReLU}(\lceil Qx \rceil) / Q$, a nondifferentiable approximate of ReLU. \textbf{Mid}: Performance of the model at different $Q$. In contrast, SGD can only train the last layer and cannot learn meaningful features from data. \textbf{Right}: gradient alignment to the ReLU network dual. As the theory predicts, the network self-assembles to a state that approximates the gradient of an approximate differentiable network, which explains the emergent learning behavior of the model. The colors denote layers 1 (orange), 2 (blue), and 3 (red), respectively.}
    \label{fig:step relu}

    \centering
    \includegraphics[width=0.9\linewidth]{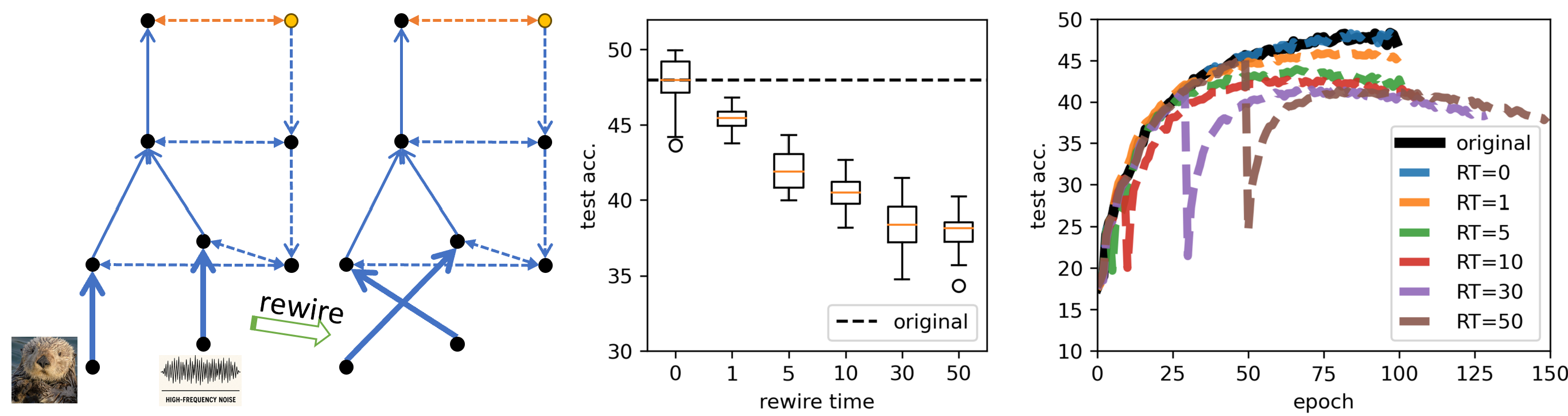}

    \caption{\small Training of heterosynaptic circuits after a rewiring of two heterogeneous multimodal channels. \textbf{Left}: Architecture of the network. The input layers are divided into two channels, one taking a visual input for which the circuit needs to learn to classify, the other taking an auditory input, which is a high frequency noise. We train the original circuit for a fixed epochs, which we call ``rewire time (RT)," and then we rewire the visual input to the auditory channel and vice cersa, and train for another 100 epochs. \textbf{Mid}: for a small RT, rewiring has no significant effect on the performance, and the network is fully capable of assemble itself, indicating the flexibility and robustness of the network. As RT increases, the performance drops, this is consistent with the well-observed phenomenon of critical periods in biological brains. \textbf{Right}: Median of 20 training trajectories. For a large RT, the performance drops by a large amount immediately after rewiring.}
    \label{fig:rewiring}
\end{figure}

\subsection{Heterosynaptic Stability Leads To Meta-Gradient Learning}

That the brain circuit may reach a stable state after an external stimulus is a form of plasticity homeostasis \cite{turrigiano2000hebb, turrigiano2004homeostatic}. One can show that for those input $x$ that this dynamics is stable (Theorem~\ref{theo: HS}): $\bar{h} =  \phi \bar{V}^\top \nabla_{p} \ell$.
We refer to $\phi$ (and its generalizations, see below) as the \textit{consistency score}. Thus, the signal $p$ receives is $p(t')=\phi H \nabla_{p(t)} \ell(p(t))$, where $H = \bar{V} \bar{V}^\top$ is a PSD matrix learning rate. Because Theorem~\ref{theo: HS} only requires one of the two heterosynaptic weights to be stable, it also applies to the case where the other weight is nonplastic.

Thus, if $\phi<0$, this node performs gradient descent. Otherwise, it performs gradient ascent. Gradient ascent may not be bad because it may help forgetting. Also, $\ell$ can be any energy function, and so the algorithm can also run energy-based models such as a Hopfield net \cite{song2021train} by updating the activations through feedback loops. Prior theories on feedback-type algorithms argued that the instructive synapses $\bar{V}$ should be a unitary matrix because they neglected the possibility of $\bar{V}$ being meta-plastic \cite{nokland2016direct, refinetti2021align}. In our framework, because $H$ is plastic, this dynamics is consistent with metaplasticity, where the plasticity of the neuron depends on prior experience. We show these metaplasticity curves and the emergence of alignment to the gradient in Figure~\ref{fig:meta-and-micro}. Because Eq.~\eqref{eq: hs} makes no assumption about other nodes -- it is thus robust to changes or lesions of other nodes. Making no assumption about the overall architecture, the algorithm is thus also fully flexible to accommodate any higher structure.

The theorem (Theorem~\ref{theo: HS advanced}) also allows $f$ to be nondifferentiable or have zero gradients everywhere as long as the activations of $f$ can be approximated point-wise by a class of activations that we call RLU. For a variant of HSP learning rule, the result extends to an arbitrary activation function class. Thus, a heterosynaptic circuit could implicitly smoothen the predictive model $f$ and thus can, in principle, ``differentiate" through neural spikes, which are believed to be nondifferentiable. We show such an experiment in Figure~\ref{fig:step relu}. Here, we train a neural network with the step-ReLU activation, a nondifferentiable and low-precision variant of ReLU. The bits of information achievable with the activation is proportional to $Q$, and we see that the performance of the model is close to optimal even when $Q \approx 4$, a very low precision. This shows that the algorithm is robust to noises or local imperfection of the activations and can operate at a low bit precision, consistent with the biological brain \cite{zheng2025unbearable, czanner2015measuring}. This is  a property that SGD does not have.

As discussed, HSP is naturally robust to changes in the model architecture. Figure~\ref{fig:rewiring} shows a simulation where we train a multimodal model with two input channels taking in visual and auditory input. We exchange and rewire the two channels after training for a period of time. We see that the model learns equally well before and after rewiring if we rewire early in training, mimicking ferrets that undergo a similar rewiring by surgery \cite{von2000visual}. The existence of a critical period of learning is also consistent with biological observations \cite{cisneros2020critical}.


\paragraph{Heterosynaptic Rule with Emergent Hebbian Dynamics} The rule \eqref{eq: update rule} may seem agnostic to whether the update is homosynaptic or heterosynaptic. This is deceptive because homosynaptic updates cannot reach a stationary state. Consider a homosynaptic scenario where $p(t)= p(t') = W \tilde{h}(t)$. Now, the dynamics becomes $\Delta {W} = W \tilde{h}(t)  h^\top(t) - \gamma \bar{V}$, identical to the Hebbian rule \cite{hebb2005organization}. The stationary points of this dynamics is either zero or infinity -- in either case, it is not biologically meaningful. In fact, this is a well-known divergent problem of the Hebbian rule \cite{chistiakova2015homeostatic}. For this reason, the Hebbian rule is usually replaced by Oja's rule, which enforces a normalization of the weights \cite{oja1982simplified}. An alternative to the Oja rule is to assume that in Eq.~\eqref{eq: update rule}, $p(t')$ must not be a direct function of $W$, and so the rule must be two-signal.

Yet, there is an interesting emergent property of Eq.~\ref{eq: update rule}: it is hard to distinguish it from the Hebbian rule. After one step of Hebbian update, we have $\Delta p(x) \propto  p(x)$. Thus, the activation becomes strengthened. When updating using Eq.~\eqref{eq: update rule} and at the HS, $\Delta p = \bar{V}\bar{h} \propto p$, which is proportional to the Hebbian update. 
Numerical results show that Hebbian dynamics dynamics do emerge in most of layers. Thus, Hebbian plasticity and heterosynaptic plasticity may be two faces of the same coin, and HBP and HSP may occur simultaneously. See Figure~\ref{fig:meta-and-micro}. There is an interesting three-phase dynamics: during training, layers are Hebbian at the beginning, then transitions to anti-Hebbian, and finally comes back to Hebbian, and later layers tend to have better overlap. This transition between Hebbian and anti-Hebbian dynamics is not well understood but may happen in animals \cite{koch2013hebbian, froemke2005spike, sjostrom2006cooperative}.

\subsection{Dynamical Consistency Is Achievable and Emergent}
We say that two nodes are consistent if they are both performing gradient ascent or descent. Namely, two nodes $i$ and $j$ have \textit{dynamical consistency} if $\phi_i \phi_j > 0$. This definition implies a principle to construct large-scale neuronal gradient systems: \textit{make neuron pairs consistent}. It turns out that it may be easy to ensure consistency. The next example shows that all layers of any deep ReLU network are always mutually consistent.

\paragraph{DC by Architecture Design.} Consider a deep ReLU network: $f(x) = W_D R_{D-1} W_{D-1} ..  R_{1}W_1 x := \hat{y}(x)$
where $R_j(x)$ is the hidden activation of $j$-th layer, which is a diagonal zero-one matrix for a ReLU network. Define $p_i$ to be the $i$-th layer preactivation, which we regard as a single node. We have that $R_j(x) =  R_j(h_{j}(x))$. For a zero-one matrix $R$, we have that for an arbitrary layer $j$
\begin{equation}
    \nabla_{p_j}^\top \ell p_j  = \underbrace{\nabla_{\hat{y}}^\top \ell  W_D...  R_j}_{\nabla_{p_j}^\top \ell}  \underbrace{W_j  ... R_1W_1 x}_{p_j} = \nabla_{\hat{y}}^\top \ell f(x),
\end{equation}
which is independent of $j$. This means that $\nabla_{h_j}^\top \ell h_j  =\nabla_{h_i}^\top \ell h_i$ for all pairs of $i$, $j$. Namely, the gradient for a ReLU network is fully consistent.

This result can be generalized to essentially any network graphs (Theorem~\ref{theo: dynamical consistency}): for any forward architecture, one can partition the forward neurons into nodes such that every partition has the same consistency score. This means that for any predictive circuit, there always exists a backward pathway that can train it! Because a dense connectivity between nodes gives Eq.~\eqref{eq: hs}, it could imply a strategy that the biological brain uses during development. For example, the new born's brains are known to have sparse connections at birth and new connections between neurons rapidly grows upto a few years old \cite{kolb2017principles}.


\paragraph{DC by Emergence.} Essentially, DC is a local property -- any node only has to be consistent with its neighbors, which are in turn consistent with their neighbors. Our simulations show that DC can emerge automatically. 
See Figure~\ref{fig:meta-and-micro} for the emergence of DC in a network with gelu activation \cite{ramachandran2017searching}, which does not satisfy DC by design. At the beginning of training, there is no consistency across different layers, and each layer could have a positive or negative consistency score before training. After training, all layers emerged to become consistent for all randomly sampled connectivity structures. In theory, we prove that if both $W$ and $\bar{V}$ reach HS, all pairs of neurons can emerge to become mutually consistent for a class of connectivity graphs (Theorem~\ref{theo: dc}).


Now, suppose that through clever wiring or evolutionary mechanisms, we have achieved dynamical consistency for all nodes. It now only suffices to set the overall sign of learning correct -- because of the dynamical consistency, when one node have $\phi<0$, all the nodes perform gradient learning. A special node that is achieves this initial condition is called the root node. The discussion above suggests the following idealized definition of the root node: we say that $\bar{h}$ is a root node if it connects to the output node $\hat{y}$ through $\bar{V}$ and  $\bar{h} = c \bar{V}^\top \nabla_{\hat{y}}\ell$ for some $c<0$. Computing the initial signal $\nabla_{\hat{y}}\ell$ is easy: for a quadratic loss, this term is just $\hat{y} -y$. In Appendix~\ref{app sec: theory}, we show how if one can compute $\nabla_{\hat{y}}\ell$, a local circuit with HS can simulate a root node. In the neocortex, the root node may be realized by the prefrontal cortex \cite{matsumoto2007medial, fuster2015past}.


\begin{figure}
    \centering
    \includegraphics[width=\linewidth]{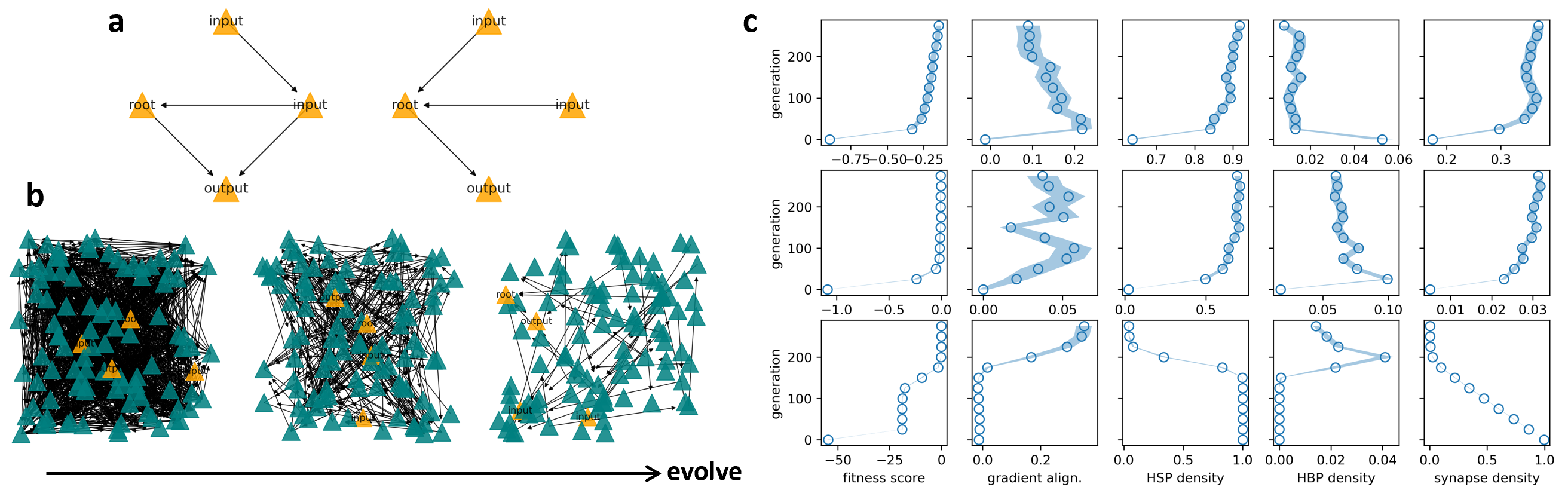}
    \caption{\small An evolutionary growth of heterosynaptic two-signal circuits. \textbf{a}: Examples of evolved circuits. The system consists of $4$ neurons whose target is to learn a linear regression problem, and the edges are pruned and grown according to a simple evolutionary algorithm for 200 generations. \textbf{b}: Example of an evolution trajectory for a densely initialized 100-neuron network. \textbf{c}: Results averaged over $400$ independent simulations, and the shaded region shows uncertainty.  See Appendix~\ref{app sec: exp} for details and another simulation with 100 neurons. The rows show 4-neuron evolution from a sparse init. (upper), 100-neuron evolution from a sparse init. (mid) and dense init (bottom). As evolution happens, the update rule of the circuit becomes increasingly aligned with the gradient.}
    \label{fig:synapse growth}
\end{figure}

\subsection{Evolution Enables HSP Gradient Learning}
 
Two-signal algorithms can be integrated with developmental or evolutionary mechanisms. If nature does utilize HSP to perform gradient learning, then, the theory implies ways for evolution to incrementally improve on itself and search for a better algorithm. Any connectivity pattern for which better dynamical consistency or stability can be reached is likely to be preferred. For example, the following are potential targets for evolution: (1) better choices of activation functions, (2) better connectivity structures, and (3) better locations of the root node. Nature can also search beyond the simple linear heterosynaptic rules given in Eq.~\eqref{eq: update rule} (Appendix~\ref{app sec: algo generalization}). We design a simple evolution experiment in Figure~\ref{fig:synapse growth}. The system operates under simple rules: (1) every neuron fires at most twice, (2) the synapse weight is updated by the product of its first and last activation, where $\Delta W \propto h_{\rm first} h_{\rm last}^\top$, (3) the overall density of synapses tend to stay at a sweet spot that is not too dense or sparse. Note that this rule allows for both Hebbian plasticity and heterosynaptic plasticity. If a neuron only fires once, its incoming weights are updated using the Hebbian rule (see Appendix~\ref{app sec: evolution}). For simplicity, we experiment with a four-neuron system, and the error neuron fires at value $\hat{y}-y$ once the output neuron receives its first output $\hat{y}$. The first generation starts from a very sparse activation and evolves according to the fitness score, which is the negative loss function after 100 training steps. The results show interesting results: (1) with a sparse init., the neurons learn to densify very quickly, similar to the developmental stage of the brain, and (2) the fraction of neurons updated with HSP increases from 0.5 to 0.9, while that of HBP decreases significantly; (3) with a dense init., the evolution leads to a steady tendency of sparsification, and gradient learning emerges suddenly as the model becomes sparser; (4) for the 100-neuron system, the HBP density first increases and drops, which suggests a two-phase dynamics. This could be because the first phase is simply due to random growth, and the second phase prunes the Hebbian edges that destabilize the system. This shows cases of both the possibility of evolving such an algorithm and the fact that HSP is easily found to be favorable compared to HBP. 

In another simulation, we macroscopically perturb the connectivity maps of the SAL algorithm (Appendix~\ref{app sec: macroscopic}). We find that if the network satisfies DC by design, the variation in the performance is much smaller than a network that does not satisfy DC by design. This could imply that the connectivity structure is important for a general nonlinearity. For example, the nodes closer to the output perform better when the instructive signal comes from a node closer to the root node.




\section{Discussion}

Our theory can be seen as a first-order theory, which shows that {\it almost any heterosynaptic circuit can compute some form of gradient descent}. It is intuitive and empirically observed that different heterosynaptic circuits behave differently as they implement different gradient algorithms. A future direction is thus to establish a ``second-order" theory that predicts which type of connectivity is better thaitn others and compares such circuits with biological anatomy. Also, we showed that heterosynaptic stability is sufficient to ensure gradient learning, but it remains unclear whether it is necessary. Our pilot simulations suggest that they may not always be required, and future theory may relax the HS condition. Also, the mechanism for the emergence of DC is yet to be understood. We only explored the simplest ways for DC to emerge; the biological brains, meanwhile, have a large repertoire for global synchronization \cite{raut2021global}. For example, mechanisms beyond HS may exist to couple the consistency scores so that they interact like a percolating system.

\vspace{-1mm}
\paragraph{How does biological learning happen?} Until now, the primary paradigm of biological learning and memory has been Hebbian homosynaptic plasticity, and the role of HSP has been regarded as auxiliary. Because HBP dynamics is a positive feedback, it leads to exploding runaway dynamics, and HSP is required to stabilize the system \cite{chen2013heterosynaptic, chistiakova2014heterosynaptic, zenke2017hebbian}. However, this paradigm of thought is questionable: on the one hand, HBP has not been demonstrated to perform well in learning tasks, and HBP seems to easily lead to exploding dynamics that requires additional stabilizing mechanisms, which are not easy to implement. 
If nature uses HBP for learning, it seems to be making a bad investment: poor performance with many downsides. 

If we switch perspective, this problem can be magically fixed: HSP is the primary learning and memory mechanism, while HBP is an emergent byproduct (maybe also auxiliary). In this perspective, HSP simulates gradient descent and looks like it performs some Hebbian learning. Because there is no direct positive feedback, there is no runaway dynamics, and, thus, no need for a standalone stabilizing mechanism. A simple decay term is sufficient to ensure stability. In reality, active homosynaptic plasticity and the emergent Hebbian plasticity may coexist in the brain, but which one is more dominant \textit{in vivo} is yet unclear.

\vspace{-1mm}
\paragraph{Stochastic Gradient Descent in the Brain} Existing works on heterosynaptic circuits suggests that to identify these heterosynatpic circuits, one can study the anatomy of the brains, and try to identify the particular shapes and motifs proposed by these algorithms \cite{akrout2019deep, liao2024self, lillicrap2016random, lillicrap2020backpropagation}. However, our result implies that little can be inferred from anatomy beyond the (critical) establishment of local heterosynaptic connections because our theory shows that heterosynaptic motifs can implement gradient in a highly flexible way. On the good side, if the brain really performs a form of gradient descent, the literature on gradient learning in AI systems becomes useful for understanding learning in brain systems, and the phenomenology they discover can serve as signatures of brain implement gradient learning. For example, it might become possible to study the kernel and feature learning regimes of the brain \cite{chizat2018lazy, yang2020feature}. The tendency to learn orthogolized representations has been found both in brains \cite{sun2025learning} and neural networks \cite{ziyin2024formation}.

\vspace{-1mm}
\paragraph{Optimizer Architecture Design} So far, the design of AI optimizers is primarily through engineering the dynamics of training by manipulating the gradient, such as adding momentum or normalization to it \cite{journals/corr/KingmaB14_adam}. Our theory -- grounded in brain implementations -- suggests a rather different perspective: optimizing a feedforward network requires adding some form of feedback connections, by (1) simply expanding the original network independently from its specific task and (2) specifically adding feedback connectivity to them. Thus, we can now design the optimizer circuit like how we design neural network models, and its architecture does not have to be similar to the predictive circuit. The trainer architecture can be arbitrary: random, trained along with the predictive model, or frozen from a pretrained model. As our numerical results show, the best-performing model trainer is not necessarily one that resembles the predictor. Extensive evidence exists to show that HSP can work surprisingly well in deep learning; for example, distillation can be seen as a form of the heterosynaptic circuit which works well in training small networks from large models \cite{hinton2015distilling}.

\vspace{-1mm}
\paragraph{Physical Learning in Analog Computers} There is a resurgence of interest in developing analog AI hardware that directly implements neural networks such as quantum neural networks and photonic networks \cite{beer2020training, ashtiani2022chip, bogaerts2020programmable}, which may be much faster or more energy-efficient than existing GPU technology. However, implementing the learning algorithm in analog computers has been challenging because it is not easy to compute explicit gradients in analog. Our theory implies a workaround -- using heterosynaptic circuits removes the need to compute gradient explicitly. Our proposal may be implemented in these physical systems as an indirect but efficient way to implement gradient descent in analog hardware.

\bibliographystyle{plain}

\clearpage
\appendix
\section{Matrix Learning Rate}\label{app sec: indistinguishability}

SGD can be generalized to have a matrix learning rate. Let $H$ be any positive semidefinite (PSD) matrix and update the parameters $\theta$ by $\dot{\theta} = - H \nabla_\theta \ell$. Then, this dynamics will lead to a monotonic decrease in $\ell$ for any PSD $H$: $\dot{\ell} = - (\nabla^\top \ell) H  (\nabla \ell)$. Well-known training algorithms such as Adam \cite{journals/corr/KingmaB14_adam}, rmsprop \cite{Tieleman2012_rmsprop}, and natural gradient descent \cite{Amari:1998:NGW:287476.287477} have such a matrix learning rate. This can be seen as the a generalized version of GD. Our theory will show that Heterosynaptic rules leads to a gradient learning dyanmics with a PSD matrix learning rate. When $H$ is full-rank, this dynamics will have identical stationary points as GD. This also means that GD and GD with matrix learning rate are difficult to distinguish experimentally, especially when there is some noise in the gradient and the only available metric is the correlation between the two.

Many types of dynamics are difficult to distinguish, especially in the presence of strong noise \cite{welvaert2013definition}. For example, there is an important alternative experiment-driven interpretation of the matrix learning rate: dynamics with different matrix learning rates are difficult (if not impossible) to distinguish experimentally. Assuming that we want to show that a group of neurons $h$ follows a dynamics of the form $K(f)$. One necessarily has to measure the empirical change $\tilde{\Delta} h$ and compute the population average (in addition to some time average) of $h$: $\hat{m} = \sum_i \tilde{\Delta} h_i K_i(h)$. Due to strong noises and variations among the neurons, one concludes that $h$ approximately follows $K(h)$ if $\hat{m}$ is significantly (and even if weakly) positive. The necessity of making rather crude measurements implies that it is difficult to distinguish any rule $\Delta h$ that is positively aligned with $K(f)$. Therefore, we might as well regard all such dynamics as a family of equivalence classes. In this sense, all matrix-learning-rate gradient rules belong to the same equivalence class that is difficult to experimentally distinguish because the learning rate is PSD.

\section{Neuron Gradient Descent}\label{app sec: neuron gradient descent}

Consider two nodes of neurons $h_a$ and $h_b$ connected by a dense synaptic matrix weight $W$: $p_b = Wh_a$, where $W$ is the learnable weight, $p_a$ is the postactivation of node $a$ and $p_b$ the preactivation of node $b$. The learning objective (or, the loss function) $\ell=\ell(h_b)$ is a function of $h_b$. The gradient for the weights are 
\begin{equation}
    \nabla_W \ell = \nabla_{p_b} \ell h_a^\top,
\end{equation}
which is the outer product of $h_a$ with the gradient of $p_b$. Thus, if one knows the neuron gradient $\nabla_{p_b} \ell$, $\nabla_W\ell$ can be computed easily with the chain rule. We thus abstract away the weight gradient and focus on the activation gradient $\nabla_{p_b} \ell$. Neuron gradients are arguably more useful because they can be used in a feedback loop to enhance local neuronal computation and serve modulatory purposes other than learning. For example, the cortical feedback, argued to be a biological circuit for learning \cite{lillicrap2020backpropagation}, serves an immediate function of modulation \cite{hultborn2001state}.
\section{Theory}\label{app sec: theory}

\subsection{Heterosynaptic Stability}
Here, we show that heterosynaptic stability ensures that neuron response to specific signals takes a special form that could be regarded as a gradient. 

The algorithm works even if the activations of the predictive model $f$ are not differentiable (or have a zero gradient everywhere). We will prove this later in the section. We first present a proof for the case when $f$ is differentiable -- and the two proofs can be compared to improve understanding. Since this theorem is more for understanding, we ignore all factors of $O(\epsilon)$ here. Because $p$ and $h$ and $\bar{h}$ are all functions of the input signal $x$, we omit specifying the dependence of these quantities on $x$.

\begin{theorem}\label{theo: HS}
   (Heterosynaptic Stability) Let $\ell(p(t))$ and $\ell'(\bar{h}(t'))$ be separate loss functions for $h$ and $\bar{h}$, respectively. For any $x$ such that $\Delta \bar{V} =0$ in Eq.~\eqref{eq: update rule},
    \begin{equation}\label{app eq: barh}
        \bar{h}(t') =  \phi_p \bar{V}^\top \nabla_{p(t)} \ell,
    \end{equation}
    \begin{equation}
        p(t) =  \phi_{\bar{h}} \bar{V} \nabla_{\bar{h}(t')} \ell',
    \end{equation}
    where
    \begin{equation}
        \phi_p  =  \frac{\gamma}{\nabla_{p(t)}^\top \ell p(t)} \in \mathbb{R},
    \end{equation}
    \begin{equation}
         \phi_{\bar{h}} = \frac{\gamma}{\bar{h}^\top(t') \nabla_{\bar{h}(t')}\ell' }. \in \mathbb{R}. 
    \end{equation}
\end{theorem}
In the proof, the quantity $\bar{h}$ is always evaluated at time $t'$, and $p$ is always evaluated at time $t$, so we omit specifying these arguments. 

\begin{proof}
    We first prove the first equation. By definition, we have 
    \begin{equation}
        \Delta \bar{V} = p \bar{h}^\top  - \gamma \bar{V}.
    \end{equation}
    At stationarity,
    \begin{equation}\label{app eq: proof 1}
        p \bar{h}^\top  =  \gamma \bar{V}.
    \end{equation}
    Multiplying $\nabla_{p}^\top \ell$ from the left,
    \begin{equation}
       \nabla_{p}^\top \ell p  \bar{h}^\top =  \gamma \nabla_{p}^\top \ell \bar{V}.
    \end{equation}
    Therefore,
    \begin{equation}
       \bar{h} =  \phi \bar{V}^\top  \nabla_{p} \ell.
    \end{equation}
    where 
    \begin{equation}
        \phi_p  = \frac{\gamma}{\nabla_{p}^\top \ell p}.
    \end{equation}
    Now, we prove the second equation. We start from Eq.~\eqref{app eq: proof 1}. Multiplying $\nabla_{\bar{h}}\ell'$ from the right, we obtain
    \begin{equation}
        p \bar{h}^\top \nabla_{\bar{h}}\ell'  =  \gamma \bar{V} \nabla_{\bar{h}}\ell'.
    \end{equation}
    Thus,
    \begin{equation}
         p   =  \phi_{\bar{h}} \bar{V} \nabla_{\bar{h}}\ell',
    \end{equation}
    where 
    \begin{equation}
        \phi_{\bar{h}} = \frac{\gamma}{\bar{h}^\top \nabla_{\bar{h}}\ell' }.
    \end{equation}
    This finishes the proof.
\end{proof}
Because regularization needs to balance with gradien, close to stationarity, the quantity $\phi$ should be of order $O(1)$, and so the effective learning is asymptotically independent of $\gamma$ \cite{ziyin2024formation}. An interesting aspect of the theorem is that it implies the emergence of a symmetric structure where $p(t)$ becomes a gradient for $\bar{h}(t')$ and vice versa simultaneously. Note that the two equations can be alternatively written as (assuming that the only target neurons of $\bar{h}$ is $p$):
\begin{equation}
    p(t') =  \phi_{p(t)} \bar{V}\bar{V}^\top \nabla_{p(t)} \ell,
\end{equation}
\begin{equation}
    p(t) =  \phi_{\bar{h}} \nabla_{p(t')} \ell' =  \phi_{p(t')} \nabla_{p(t')} \ell'.
\end{equation}

The HS condition for the weight $W$ is exactly the same. Due to symmetry in the dynamics for $W$ and $\bar{V}$, by applying theorem~\ref{theo: HS}, one trivially obtains the following theorem.
\begin{theorem}\label{theo: HS II}
   (Heterosynaptic Stability II) Let $\ell'(p(t'))$ and $\ell(\tilde{h}(t))$ be separate loss functions for $p(t')$ and $\tilde{h}(t)$, respectively. For any $x$ such that $\Delta W =0$ in Eq.~\eqref{eq: update rule},
    \begin{equation}
        \tilde{h}(t) =  \phi_{p(t')} W^\top \nabla_{p(t')} \ell',
    \end{equation}
    \begin{equation}
        p(t') =  \phi_{\tilde{h}(t)} W \nabla_{\tilde{h}(t)} \ell,
    \end{equation}
    where
    \begin{equation}
        \phi_{p(t')}  =  \frac{\gamma}{\nabla_{p(t')}^\top \ell p(t')} \in \mathbb{R},
    \end{equation}
    \begin{equation}
        \phi_{\tilde{h}(t)}  =  \frac{\gamma}{\nabla_{\tilde{h}(t)}^\top \ell \tilde{h}(t)} \in \mathbb{R}. 
    \end{equation}
\end{theorem}

Similarly, these equations imply the following equations (assuming that $p(t)$ are the only target neurons of $\tilde{h}(t)$):
\begin{equation}
    p(t) =  \phi_{p(t')}  W W^\top \nabla_{p(t')} \ell',
\end{equation}
\begin{equation}\label{app eq: pt'}
    p(t') =  \phi_{p(t)} W W^\top \nabla_{p(t)} \ell.
\end{equation}

In the main text, we have focused on discussing Eq.~\eqref{app eq: barh}, which is relevant to the case where both $W$ and $\bar{V}$ are plastic (assuming that our goal is to train $W$). However, when $\bar{V}$ is stationary, this equation is no longer relevant, and one needs to look at Eq.~\eqref{app eq: pt'} instead, which essentially states the same result as Eq.~\eqref{eq: hs}, where $WW^\top$ is the matrix learning rate.

\paragraph{Nondifferentiable $f$} When $f$ is not differentiable, we need to define an approximate model $F$ that has the same graph structure as $f$ but has only differentiable (and well-behaved) activations that approximate $f$. Let us label all the latent nodes of $f$ by subscript $i$. Let ${F}$ be any neural network with latent layers $\zeta_i$ such that for any $x$ and any $t$
\begin{equation}
    \zeta_i(x, t) = p_i(x, t) + O(\epsilon).
\end{equation}
Then, one can prove the following theorem. The second equation in Theorem~\ref{theo: HS} can be proved in exactly the same way, so we omit it for brevity.
\begin{theorem}\label{theo: HS advanced}
   (Heterosynaptic Stability) For any $x$ such that $\Delta \bar{V}_i = O(\epsilon)$,
    \begin{equation}
        \bar{h}_i =  \phi_i \bar{V}_i^\top \nabla_{\zeta_i(t)} \ell(F(x,t)) + O(\epsilon).
    \end{equation}
    where
    \begin{equation}
        \phi_i  =  \frac{\gamma}{\nabla_{\zeta_i}^\top \ell(F(x,t)) \zeta_i(x,t)} \in \mathbb{R}.
    \end{equation}
\end{theorem}
Similar to the previous proof, $p_i$ and $\zeta_i$ are evaluated at time $t$, and $\bar{h}_i$ is evaluated at time $t'$. We thus omit $t$ and $t'$ from the proof.

\begin{proof}
    By definition, we have 
    \begin{equation}
        \Delta \bar{V}_i = p_i(x) \bar{h}_i^\top(x)  - \gamma \bar{V}_i.
    \end{equation}
    Close to stationarity,
    \begin{equation}
        p_i(x) \bar{h}_i^\top(x)  =  \gamma \bar{V}_i + O(\epsilon).
    \end{equation}
    Multiplying $\nabla_{\zeta_i}^\top \ell$ from the left,
    \begin{equation}
       \nabla_{\zeta_i}^\top p_i \ell  \bar{h}_i^\top(x) =  \gamma \nabla_{\zeta_i}^\top \ell \bar{V}_i + O(\epsilon).
    \end{equation}
    Therefore,
    \begin{equation}
       \bar{h}_i(x) =  \phi_i  \bar{V}_i^\top  \nabla_{\zeta_i} \ell + + O(\epsilon).
    \end{equation}
    where 
    \begin{equation}
        \phi_i = \frac{\gamma}{\nabla_{\zeta_i}^\top \ell \zeta_i(x)}.
    \end{equation}
\end{proof}

\begin{remark}
    Note that the dynamics of $\Delta \bar{V}$ can be completely {nonstationary}. The only condition required is that $\Delta \bar{V}$ is small, which can happen for some $x$ just due to randomness, and the signal coming from other $x$ may just contribute an essentially random noise signal that cancels out on average, which is what numerically results seem to imply.
\end{remark}

The proof shows that the model will learn to update $f$ in a way as if it is performing gradient descent training for $F$. If $f$ is fully differentiable, one can simply replace $F$ with $f$ and $\zeta$ with $h$ to obtain a simpler proof, which we have presented in the first part of this section. In exactly the same way, by making these replacements, all the theorems we prove in the next section can also be extended to nondifferentiable $f$. This is a exercise which we leave to the readers.

\clearpage

\subsection{Dynamical Consistency by }
We will make the following assumption for the loss function.
\begin{assumption}\label{assump: separation}
    (Path separation implies gradient separation) For any two sets of neurons in the network ($h_1,\ h_2$), if there exists constant matrix $Z >0$ such that 
    \begin{equation}
        h_1 = Z h_2,
    \end{equation}
    then, there exists $c_0 >0$ such that 
    \begin{equation}
       c_0 Z^\top \nabla_{h_1} \ell  = \nabla_{h_2} \ell.
    \end{equation}
\end{assumption}
This result is trivially true for feedforward networks, or if the loss function depends on $h_2$ only through its dependence on $h_1$.

\subsubsection{Emergence of DC}
We consider a class of nonlinear activations that can be described by the following definition.
\begin{definition}\label{def: rlu}
    An activation $h(p)$ is said to be a radially linear unit (RLU) if 
    \begin{equation}
    h = D(p) p,
    \end{equation}
    where $D$ is a diagonal matrix and $D_{ii}(p)$ is an arbitrary piecewise constant function.
\end{definition}
 Two common activation functions that obey this form is ReLU and Leaky-ReLU. This class of functions is almost everywhere continuous and differentiable. Also, note that almost any activation functions can be approximated by a member of this class (although usually in a discontinuous way). It is also possible to generalize the nonlinearity to a nondiagonal $D$.

At time $t$, suppose a set of neurons  $\tilde{h}$ fire, which then leads to another set of neurons $h$ to fire, and so
\begin{equation}
    p(t) =  \tilde{W} \tilde{h}(t),
\end{equation}
\begin{equation}
    h(t) = D(p(t))p(t),
\end{equation}
where $\tilde{W} =M \odot W$, for plastic weights $W$ and a fixed zero-one connectivity matrix $M$.

The following theorem states that all three quantities must be mutually consistent. Since all events happen at time $t$, we omit the notation $t$ in the proof.
\begin{theorem}
    Let the target function $\ell$ depend on a set of neurons viewed as a vector $h$, $\ell= \ell(h)$. Then,
    \begin{equation}
        \phi_h = \phi_p = \phi_{\tilde{h}}.
    \end{equation}
\end{theorem}
\begin{proof}
    By definition, 
    \begin{equation}
        \nabla_p h = D(p).
    \end{equation}
    Therefore, 
    \begin{equation}
        \phi_h= \nabla_h^\top\ell(h)  h = \nabla_h^\top\ell(h) D(p) p = \nabla_p^\top\ell(p) p = \phi_p.
    \end{equation}
    Similarly, 
    \begin{equation}
        \phi_p = \nabla_p^\top\ell(p) \tilde{W} \tilde{h} = \nabla_{\tilde{h}}^\top\ell(\tilde{h})  \tilde{h} = \phi_{\tilde{h}},
    \end{equation}
    where we have used assumption~\ref{assump: separation}.
    This finishes the proof.
\end{proof}

This shows that a set of neurons is consistent with its immediate input neurons, independent of the connectivity structure $M$. The next section shows that if we allow some degree of asymmetry between the update rules of $\bar{V}$ and $W$, then this mutual consistency can emerge for an arbitrary activation function.

\subsubsection{Algorithmic Generalization}

First of all, the heterosynaptic circuit can be generalized to the case where the update is not with the activations but with a nonlinear function of the activations. For example,
\begin{equation}
    \Delta \bar{V} = g_1(h) g_2(\bar{h})^\top - \gamma \bar{V}.
\end{equation}
Note that the consistency score is now of functional of $g_1$. To distinguish with the original consistency score $\phi$, we use $\kappa$ to denote the generalized consistency score:
\begin{equation}
    \kappa[g_1] = \nabla_h^\top\ell g_1(h),
\end{equation}
and so the sign of the update is determined by the quantity $g_1$. Thus, for a fixed architecture, one might be able to find a local function $g_1$ that works well. Also, searching for $g_2$ might also be useful. As an extreme example, letting $g_2(h) =\nabla_h^\top \ell h \bar{h}$ would have completely removed the dynamical consistency problem, but this is unlikely to happen for biological systems. But, still, it is possible for biology to locally optimize over the nonlinear functional forms of $g_1$ and $g_2$, given the vast amount of nonlinear biophysical factors that could play a role in influencing the plasticity of neurons.

Now, let us consider the following choice of the update rule:
\begin{align}\label{app sec: new update rules}
    \Delta \bar{V} = \eta J^{+}(p)D(p) p(t) \bar{h}^\top(t')  - \gamma \bar{V},\\
    \Delta W = \eta p(t') \bar{h}^\top(t)  - \gamma W,
\end{align}
where $D(p)$ is the activation matrix such that 
\begin{equation}
    h = D(p) p,
\end{equation}
and 
\begin{equation}
    J(p) =\nabla_p h
\end{equation}
is the Jacobian of $h$ as a function of $p$, and $J^+$ is its pseudoinverse. Namely, both the $D$ and $J$ matrices are properties of the local activation, which can in principle be approximated through nonlinear biochemical processes of the local neuron. We assume that the activation satisfies the following property: for any $p$,
\begin{equation}
    {\rm ker} D(p) \subseteq  {\rm ker}(J(p)) . 
\end{equation}
This means that gradient is zero for those neurons that do not fire. This assumption holds for ReLU, and certainly for any activation that is invertible.

Now, the HS conditions for this dynamics lead to the following two lemmas. For simplicity, we set $\eta=1$.
\begin{lemma}\label{lemma:1}
    When $\Delta W =0$,
    \begin{equation}
        p(t') = \phi_p W W^\top \nabla_{p}\ell.
    \end{equation}
\end{lemma}
\begin{proof}
    By the construction in Eq.~\eqref{app sec: new update rules}, we have that at stationarity
    \begin{equation}
         p(t') \tilde{h}^\top(t)  = \gamma W.
    \end{equation}
    Multiplying $W^\top \nabla_{p}\ell$, we obtain that 
    \begin{equation}
         p(t') \tilde{h}^\top(t) W^\top \nabla_{p}\ell  = \gamma W W^\top \nabla_{p}\ell.
    \end{equation}
    This simplifies to
    \begin{equation}
         p(t')   = \phi_p W W^\top \nabla_{p}\ell,
    \end{equation}
    where
    \begin{equation}
        \phi_p = \frac{\gamma} {p^\top \nabla_{p}\ell}.
    \end{equation}
\end{proof}

\begin{lemma}
    When $\Delta \bar{V} =0$,
    \begin{equation}
        \bar{h}(t') =  \kappa_p[g_1] \bar{V}^\top \nabla_{p(t)} \ell.
    \end{equation}
\end{lemma}

The proof is essentially the same. Note that $p(t')$ is the same but $\bar{h}(t')$ due to the introduced asymmetry. These two lemmas allow us to prove the following theorem.
\begin{theorem}\label{theo: dc}
    For any activation function $D(h)$, for any data point such that $\Delta W = 0$ and $\Delta\bar{V} = 0$,
    \begin{equation}
        \phi_p \phi_h \geq 0.
    \end{equation}
\end{theorem}
\begin{proof}
    For this choice of $g_1$, we have
\begin{align}
    \phi_p[g_1]  &= \nabla_{p}^\top \ell(p) J^{+} D p \\
    &= \nabla_{h}^\top \ell(h) J J^{+} D p\\
    &= \nabla_{h}^\top\ell(h)  Dp\\ 
    &= \nabla_{h}^\top\ell h\\
    &= \phi_h.
\end{align}

Thus,  
\begin{equation}
     \bar{h}(t') =  \phi_h \bar{V}^\top \nabla_{p(t)} \ell.
\end{equation}
and so
\begin{equation}
    p(t')= \bar{V} \bar{h}(t') =  \phi_h\bar{V}  \bar{V}^\top \nabla_{p(t)} \ell.
\end{equation}
Comparing this with Lemma~\ref{lemma:1}, we have
\begin{equation}
    \phi_h\bar{V}  \bar{V}^\top \nabla_{p(t)} \ell = \phi_p W W^\top \nabla_{p}\ell.
\end{equation}
Multiplying $ \nabla_{p}^\top\ell$ on the left, we obtain that for some $c_0,\ c_1 \geq 0$,
\begin{equation}
    c_0 \phi_h  = c_1\phi_p.
\end{equation}
This implies that 
\begin{equation}
    \phi_h \phi_p \geq 0.
\end{equation}
This finishes the proof.

\end{proof}

Together with Assumption~\ref{assump: separation}, this implies that $\phi_p$, $\phi_h$, $\phi_{\tilde{h}}$ are of the same sign, which, in turn, implies that for all nodes in the network, the consistency scores will be of the same sign. While this proof implies a strong result, it also relies on the rather strong assumption~\ref{assump: separation}. An important future theoretical problem is to relax this assumption.

\subsection{Existence of Consistent Groups}\label{app sec: algo generalization}

For notational simplicity, we prove this for the RLU type of activations and the original update rule in Eq.~\eqref{eq: update rule}. The proof can be generalized to the generalized update rule in Eq.~\eqref{app sec: new update rules}. 

\begin{theorem}\label{theo: dynamical consistency}
    (Consistency is achievable) Let $G = \{z_1,...,z_N\}$ be the neurons that the model output $f(x)$ depends on. Then, for any forward architecture, there exists (mutually nonexclusive) subsets $h_1,..., h_M \subset G$ of forward neurons into nodes such that (1) $\bigcup_{i}^M h_i = G$ and (2) dynamical consistency holds for every pair of these subsets;
\end{theorem}
\begin{proof}
    Note that because the forward computation graph is acyclic, the loss function can be decomposed into a composition of $2D-1$ functions, where $D$ is the depth of the acyclic graph:
    \begin{equation}
        \ell(x) =  L \circ K_D \circ \Sigma_D \circ ... \Sigma_1 \circ K_1 (x),
    \end{equation}
    where $L$ is the loss function for the label, $K_i$ is a linear transformation realized by multiplying a weight matrice, which we also denote as $K_i$, and $\Sigma_i$ are the nonlinearities of each layer:
    \begin{equation}
        \Sigma_i (h) = D_i(h) h,
    \end{equation}
    where the diagonal functions of $D_i$ are either a RLU nonlinearity given by Definition~\ref{def: rlu}, or a constant function equal to $1$ because there could be skip connections in the original graph. Now, defining 
    \begin{equation}
        h_i(x) = \Sigma_i \circ K_i \circ ... \circ \Sigma_1 \circ K_1 (x),
    \end{equation}
    gives the desired subsets of neurons. To see this, note that the derivative of $h_i$ is 
    \begin{equation}
        \nabla_{h_i} \ell =   K_{i+1}^\top  ... D_D^T(h_{D_1}) K_D^\top \nabla_{h_D} L 
    \end{equation}
    where we have slightly abused the notation to regard $K_D$ also as a constant matrix: $K_i(x) = K_ih$. This means that 
    \begin{equation}
        \nabla_{h_i}^\top \ell h_i =    \underbrace{\nabla_{h_D}^\top L K_D D_D(h_{D_1})... K_{i+1}}_{\nabla_{h_i} \ell} \underbrace{D_i ... D_1 K_1 x}_{h_i} = \nabla_{h_D}^\top L h_D,
    \end{equation}
    which is independent of $i$. This means that $\phi_i$ for every subset $h_i$ is mutually consistent. This division into subsets certainly includes all neurons at least once, so the proof is complete.
\end{proof}
\begin{remark}
    While the proof looks simple, it contains a lot of interesting aspects. A direct consequence is that if one wants to guarantee gradient learning, one should create auxiliary nodes that connect densely to each of these subsets of forward neurons. One interesting observation is that if there are skip connections, then some neurons will appear in multiple subsets. This means that some synapses will receive more than one signal for learning. See Figure~\ref{fig:consistent subset example} for an illustration of such grouping.
\end{remark}

\begin{figure}
    \centering
    \includegraphics[width=0.5\linewidth]{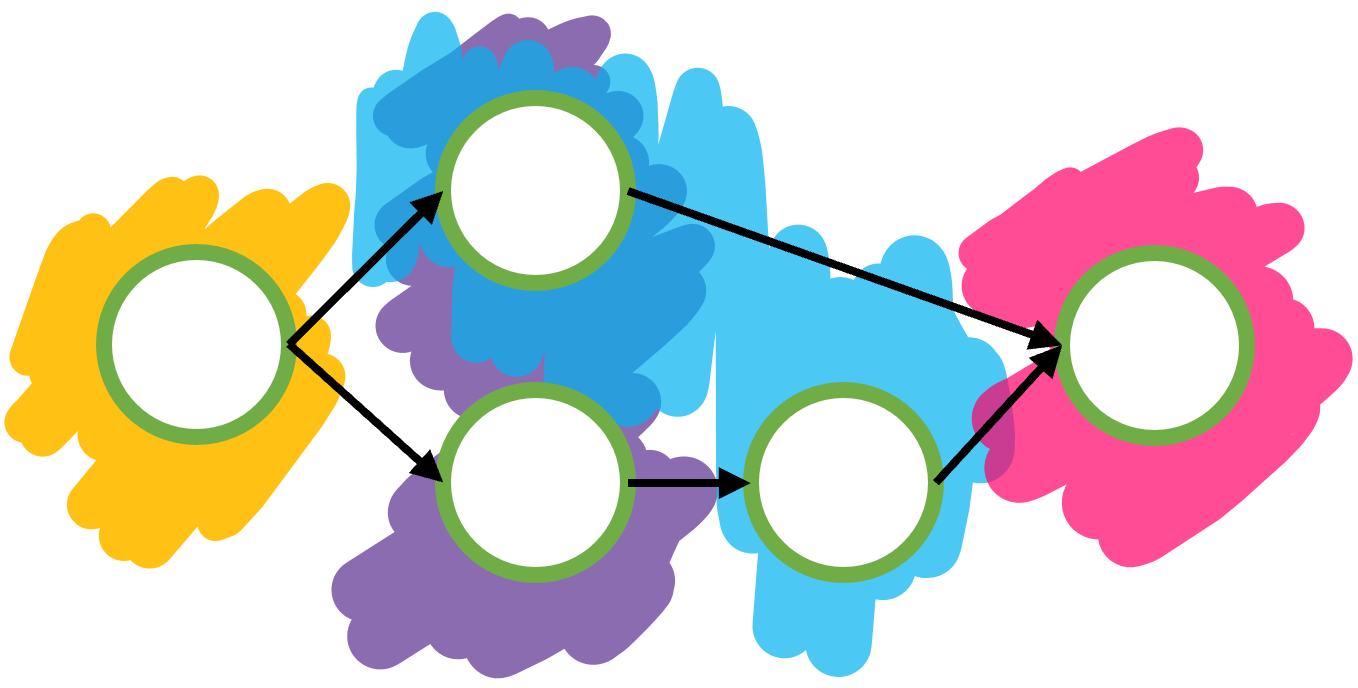}
    \caption{A network with five neurons and skip connections. Here, each coloring denotes a subset of neurons, which shows an example of how the neuron can be grouped into subsets such that every subset has the same consistency score as other ones.}
    \label{fig:consistent subset example}
\end{figure}

\subsubsection{Root Node and HS Give Rise to Gradient Descent}

Let $h$ be the activation of the output note of the network and $\bar{h}$ be the root node such that when the network computes an output $h =\hat{y}$, $\bar{h}$ computes
\begin{equation}
    -\nabla_{h} \ell(h).
\end{equation}
Now, connect $\bar{h}$ to $h$ by an identity matrix. The input weight $W$ to $h$ is updated using the HSP rule:
\begin{equation}
    \Delta W = - \gamma W + \bar{h} \tilde{h}^T,
\end{equation}
where $\tilde{h}$ is the set of all neurons connected $h$.
For those data points that this quantity is stationary, we have that 
\begin{equation}
    \gamma W = \bar{h} \tilde{h}^T,
\end{equation}
which implies that 
\begin{align}
    \gamma WW^T &= \bar{h} \tilde{h}^T W^T\\
    &= \bar{h} h^T\\
    &= -\nabla_{h} \ell(h) h^T.
\end{align}
Taking the trace of both sides, we obtain that 
\begin{equation}
    h^T \nabla_{h} \ell(h)  =- {\rm Tr}[\gamma WW^T] \leq 0.
\end{equation}
This shows that the root node is sufficient to ensure the correct sign for the updates.

\clearpage
\section{Experimental Details}\label{app sec: exp details}
This section gives all the details for the numerical simulations in this paper.

\subsection{Metaplasticity}
We consider learning a single neuron for a linear regression task:
\begin{equation}
    h = wx_t,
\end{equation}
where $x_t$ obeys i.i.d. normal distribution for every $t$. A root node computes the gradient:
\begin{equation}
    \bar{h} = h - y.
\end{equation}
The update rule obeys Eq.~\eqref{eq: update rule} with $\gamma=1$, with a learning rate of $0.1$. We initialize this $w=2$ and $v=0$ and run the algorithm for two steps (with two i.i.d. $(x,y)$ pairs), and compute its plasticity curve for different values of $h(x)$ by varying $h$. We run this process 200 times to compute the distribution of these curves. See Figure~\ref{fig:meta-and-micro}.

\subsection{Random Computation Graphs}

Because the theory implies infinitely many ways to construct a learning circuit that simulates gradient descent. It is unclear how one could make a uniform sampling from all possible circuits to test our theory. We thus regard the SAL architecture (Figure~\ref{fig:circuit examples}) as the ``canonical" architecture and sample random topologies that are either microscopic or macroscopic variations of this architecture.

We first vary the low-level connectivity structures and then the global connectivity structures. Unless stated otherwise, the experiments for this section and the next section are all done on the CIFAR-10 dataset; training proceeds for 100 epochs with a learning rate of $0.005$ and a weight decay with strength of $0.001$ for all plastic weights.

\subsection{Low-Level Connectivity Structures}

See Figure~\ref{fig:meta-and-micro}-b1 for an illustration of the entire architecture used in this experiment. Here, the network is defined to be a four-hidden layer feedforward network with sparse connectivity in all plastic weights. At time $t$, the network is in the inference phase, and at a later time $t'$, the network is in the update phase.

The whole architecture looks like a feedforward architecture, where
\begin{equation}
    p^{l+1}(t) = (M^l \odot W)h^l,
\end{equation}
where $W^l$ is the plastic weights of the $l$-th layer, and $M^l$ is a fixed zero-one matrix prescribing the connectivity pattern. Each element of the random mask is generated as a Bernoulli variable with $70\%$ probability being $1$ and fixed throughout training. The activation is chosen to be either gelu, which is differentiable but does not satisfy DC by design, or step-ReLU, which has a zero gradient everywhere.

The backward pathway is defined similarly:
\begin{equation}
    \bar{p}^{l+1}(t) = (\bar{M}^l \odot \bar{W}^l)\bar{h}^l,
\end{equation}
where $\bar{W}$ is the plastic weight and $\bar{W}$ is a random mask fixed at initialization. We let both $h^l$ and $\bar{h}^l$ to be in $\mathbb{R}^{2000}$, except for the input and output layer, whose dimensions match those of the data. Thus, each layer consists of $2000$ nodes, which, in expectation, are connected to $1400$ nodes in the next layer.

The cross-connections are also defined similarly:
\begin{equation}
    {p}^{l+1}(t') = (\bar{M}^l_V \odot \bar{V}^l)\bar{h}^l,
\end{equation}
where $\bar{V}^l$ is the plastic weight and $\bar{M}_V^l$ is a random mask fixed at initialization. The cross-connection from $h$ to $\bar{p}$ is also similarly defined. As in \cite{liao2024self}, during training, we clip the update norms at $5$ to prevent numerical instability; we also use a momentum $0.9$ on the updates to speed up training. The masking trick suggested by \cite{liao2024self} is also used to improve the performance. We run this experiment with independently sampled connectivity matrices $M$ for $100$ times to make the boxplots presented in this work.

\clearpage

\subsection{High-Level Connectivity Structures}\label{app sec: macroscopic}

\subsubsection{ReLU}
We also experiment with varying degrees of high level connectivity. Here, all connections between nodes are dense and fully connected. We first train a ReLU network with four hidden layers of forward pathway (FP) and four hidden layers of backward pathway (BP). As shown in Figure~\ref{fig:circuit examples}, the whole circuit takes a two-pathway structure. Here, we treat every layer as a node and exhaustively search for all possible cross connections. See Figure~\ref{fig:best worst circuits} for examples of such circuits. For a four hidden layer network, there are a total of $6! \times 6! = 518400$ possible graphs, which is too large a space for our limited computation power. We thus fix either the forward-to-backward connections or the backward-to-forward connections and search for the rest. This is a large enough space that includes many interesting topologies. For example, the feedback alignment and direct feedback alignment are both special cases of all edge sets we search over.

\paragraph{Backward-to-Forward Connectivity} Every node is connected to a set of other nodes with an all-to-all connection between the neurons in these nodes, and we ensure that every node receives at least one input node from the other pathway. For the first experiment, we fix all the forward-to-backward edges to be the same as that of the SAL architecture and search for all possible edge sets from the backward path to the forward pathway. For every edge set, we train the model according to the rule in Eq.\eqref{eq: update rule}. See Figure~\ref{fig:cifar10 btf relu} for the results. The zeroth node refers to the root node itself. Also, we label the nodes of the FP according to their closeness to the input layer, and we label the nodes of the BP according to their closeness to the root note.

Two observations are quite salient:
\begin{enumerate}
    \item To leading order, the performance of the majority of the models is quite close, in agreement with the expectation that gradient computation is quite universal and easy to achieve;
    \item To second order, some tendencies of connectivity structures are preferred and lead to better performances; layers prefer to be roughly aligned, as the layers closer to the output prefer being connected to the layers closer to the error signal.
\end{enumerate}
For the second point, while the best performing connectivity is something one intuitively expects, it does not take any specific form that prior works have proposed. For example, the best performance is achieved when (1) the third layer of FP is connected to the second layer of the BP, (2) the second layer of FP is connected to the second layer of the BP, and (3) the third layer of FP is connected to the third layer of the BP. See Figure~\ref{fig:best worst circuits}.

As a comparison for the need to train the backward pathway, we also run the same experiments where the weights of the backward pathway is nonplastic. See Figure~\ref{fig:cifar10 random}. We see a consistent decrease in the performance of the model by $1-2\%$, which demonstrates the advantage of full plasticity of the circuit. Interestingly, when the backward weights are nonplastic, the most preferable connectivities are those that only come from the last node of the backward pathway.

\paragraph{Forward-to-Backward Connectivity} Similarly, we fix the connectivity of the backward-to-forward connections and vary the edge set of the forward-to-backward connection. See Figure~\ref{fig:cifar10 ftb}. We see that the effect is much smaller in comparison to the effect of change in the forward activations, but the trend is similar -- the nodes closer to the output prefer being connected to the nodes closer to the root node. The existence of these systematic preferred connectivity patterns suggests the possibility of these connectivity graph being incrementally improved with simple evolutionary or developmental strategies.

\subsubsection{Sign-Activation}
Lastly, we also experiment with the sign activation to illustrate the difference between having DC by design and not. Note that the previous two sections used ReLU as activations, and the DC property is automatically satisfied. Here, for the sign activation, the DC property is no longer satisfied. See Figure~\ref{fig:cifar10 btf sign}. Here, we see that the variation in the model's performance is very large, ranging from completely trivial ($10\%$) to quite well-performing $(\sim 40\%)$. This means that network connectivity can play a significant role in the emergence of DC.

\subsection{Step-ReLU Experiments}
For these experiments, we do not search over different random connectivity graphs. We simply take the SAL architecture and sample different seeds 20 times to compute the performances. The SGD performance is the average over 20 seeds; the standard deviation of the SGD performance is less than $1\%$ and is not shown in the figure because the difference with the HSP algorithm is far larger.

\subsection{Synapse Evolution}\label{app sec: evolution}

\paragraph{Task} The task is learning a simple 2d linear regression:
\begin{equation}
    y = x_1 + 0.2 x_2,
\end{equation}
where $x_1$ and $x_2$ are independent Gaussian variables. During training, we train in an online fashion where only a single data point is seen by the network, which is the most biologically plausible sampling.

\paragraph{Model} Here, the model consist of four neurons $h \in \mathbb{R}^4$ and a trained weight $W \in \mathbb{R}^{4\times 4}$:
\begin{equation}
    h^{t+1}  = (M \odot W) h^\top,
\end{equation}
where $t$ is the step of recurrence, and $\odot$ denotes element-wise product. For every member of the population, $M$ is a fixed zero-one matrix, which denotes the connectivity graph of this network. If a neuron fires more than twice, all its activation will be set to zero -- namely, a neuron fires at most twice. Note that because each neuron can fire more than once, this model can be seen as the simplest type of recurrent network. The density of the $M$ is defined as 
\begin{equation}
    {\rm \delta}(M) = \frac{1}{16}\sum_{ij} \mathbbm{1}_{M_{ij} =1 },
\end{equation}
where $\mathbbm{1}$ is the indicator function.

\paragraph{Firing Rule} The first firing (nonzero activation value) of $h_1$ and $h_2$ are clamped to be $x_1$ and $x_2$, respectively. $h_4$ is the output neuron, and its first firing is regarded as the output of the network. Once $h_4$ fires, $h_3$ fires at the rate $\hat{y}- y$. Namely, $(h_1,h_2)$ are the input nodes. $h_3$ is the root node, and $h_4$ is the output node. We also experimented with the case where $y= x_1$ and $h_2$ is purely auxiliary (meaning that it can be used for intermediate computation but not necessarily), and the result is similar to the currently reported version.

\paragraph{Inner and Outer Loop} The experiment involves running a two-loop procedure. The outer loop evolves the mask matrix $M$ by perturbing the existing ones. The inner loop evolves the weight $W$ by training it with the two-signal algorithm. The details of these two loops are given below.

\paragraph{Evolution}  We initiate $M$ to have $20\%$ density for each member of the first generation. Every generation contains 20 members, each of which is trained for 50 steps with the two-signal algorithm (described below) to obtain the final training loss $\ell$, and the fitness score is defined as $-\ell$.

Three members with the lowest fitness score are selected as the parents of the next generation. The mask $M$ of each of these members is randomly perturbed with the following rule: (a) every element has probability $30\%$ to be perturbed, and (b) the perturbed elements are flipped from zero to one or one to zero with probability $\delta(H) - 0.05$. Namely, it is expected to have a density slightly smaller than the parent. This biases the evolution towards sparser connectivity patterns.

\paragraph{Two-Signal Learning Algorithm} During learning, the weights are updated according to:
\begin{equation}\label{eq: evolution update rule}
    \Delta W_{ij} =  0.05  \times h_i^{\rm last} h_j^{\rm first} - 0.01 W_{ij},
\end{equation}
where $h^{\rm first}$ is the first firing activation of $h$ and $h^{\rm last}$ the last. A crucial advantage of this rule is that it allows for both Hebbian update and heterosynaptic update, and we are thus letting evolution decide which update rule is better.

\paragraph{Hebbian or Heterosynaptic} In the experiment, we also measure the density of Hebbian and Heterosynaptic plasticity. While this is difficult to decide directly, it can be approximated by the following metrics, \textit{Hebbian density}:
\begin{equation}
    \delta_{\rm hebb} = \frac{1}{4}\sum_i \mathbbm{1}_{h_i^{\rm last}  = h_i^{\rm first}}  \times  \mathbbm{1}_{h_i^{\rm last}\neq 0}  \times \mathbbm{1}_{h_i^{\rm first}\neq 0}.
\end{equation}
Namely, it is the number of neurons that has fired at least once and whose last firing is identical to its first firing. According to the update rule in Eq,~\eqref{eq: evolution update rule}, if neuron $j$ fired once, its input neuron $i$ must also have fired at most once. Thus, if both $i$ and $j$ fired exactly once, it must be the case that $i$ caused $j$ to fire, and so the update rule must be Hebbian. Therefore, $\delta_{\rm hebb}$ lower bounds the number of Hebbian updates made in the network.

The \textit{Heterosynaptic density} is similarly defined as
\begin{equation}
    \delta_{\rm hetero} = \frac{1}{4}\sum_i \mathbbm{1}_{h_i^{\rm last}  \neq h_i^{\rm first}}  \times  \mathbbm{1}_{h_i^{\rm last}\neq 0}  \times \mathbbm{1}_{h_i^{\rm first}\neq 0}.
\end{equation}
Again, according to the update rule, if the firings of $i$ are not identical, this neuron must have fired twice, possibly excited by two different neurons. Now, assume that the first firing of this neuron $h_i^{\rm first}$ is excited by $j$ (possibly) in addition to other neurons, then the weight $W_{ij}$ must be updated by the Heterosynaptic rule. Therefore, this metric also lower bounds the number density of the HSP updates in the model.

\paragraph{100-neuron simulation} The model now has 100 neurons, and the connectivity matrix is now a $100$-by-$100$ matrix, having $10^4$ parameters to learn. Neurons 1-4 obey the same rules as input, output and root neurons. Neurons 5-100 are latent neurons that can serve either the purpose of inference or learning.


\clearpage
\section{Additional Experiments and Figures}\label{app sec: exp}


\begin{figure}[t!]
    \centering
    \includegraphics[width=0.3\linewidth]{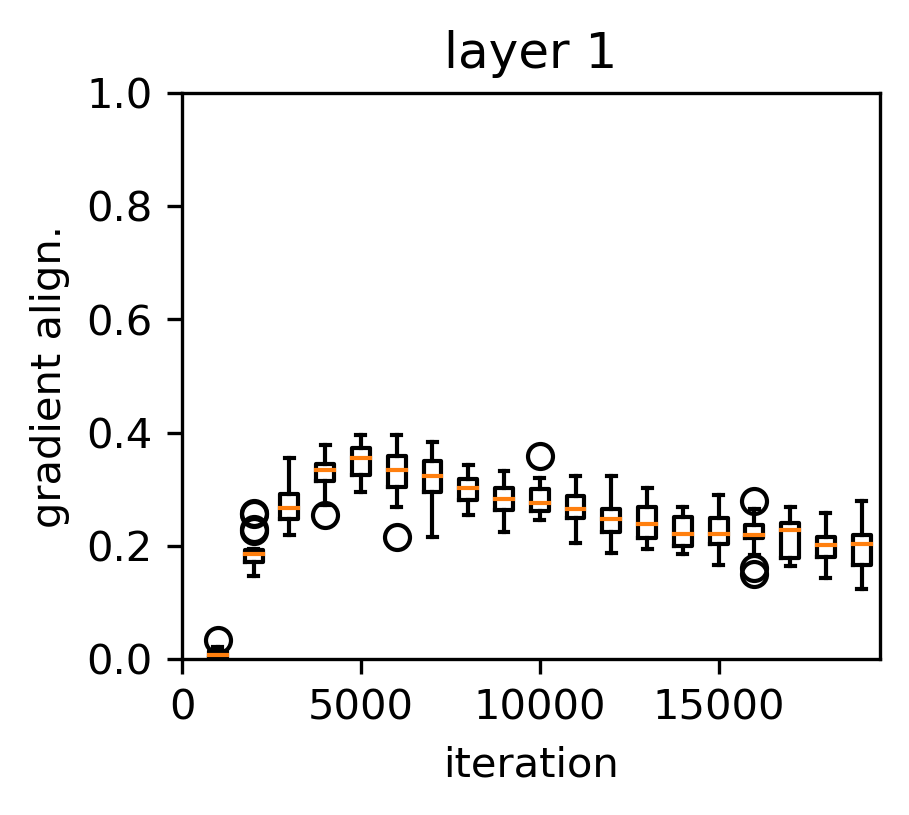}
    \includegraphics[width=0.3\linewidth]{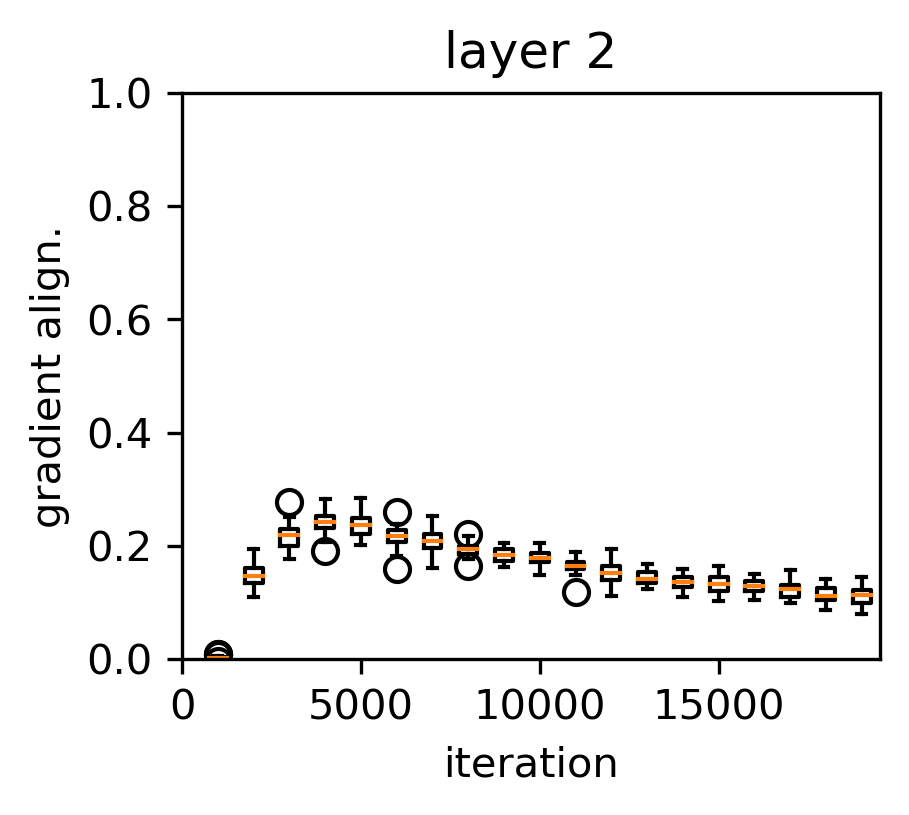}
    \includegraphics[width=0.3\linewidth]{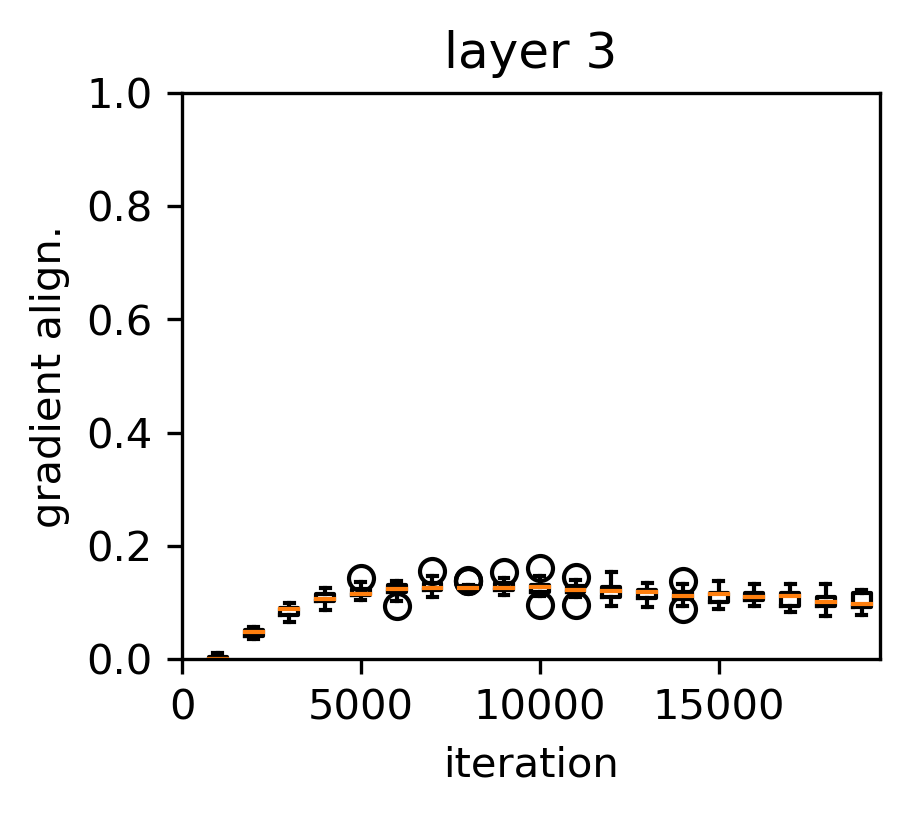}

    \includegraphics[width=0.3\linewidth]{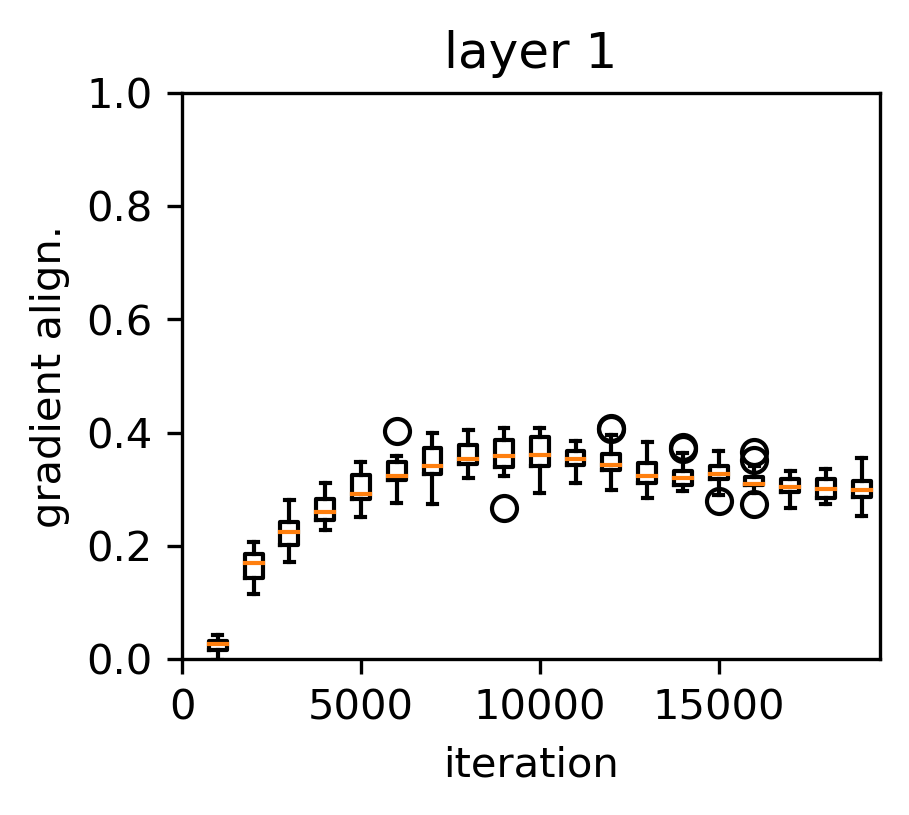}
    \includegraphics[width=0.3\linewidth]{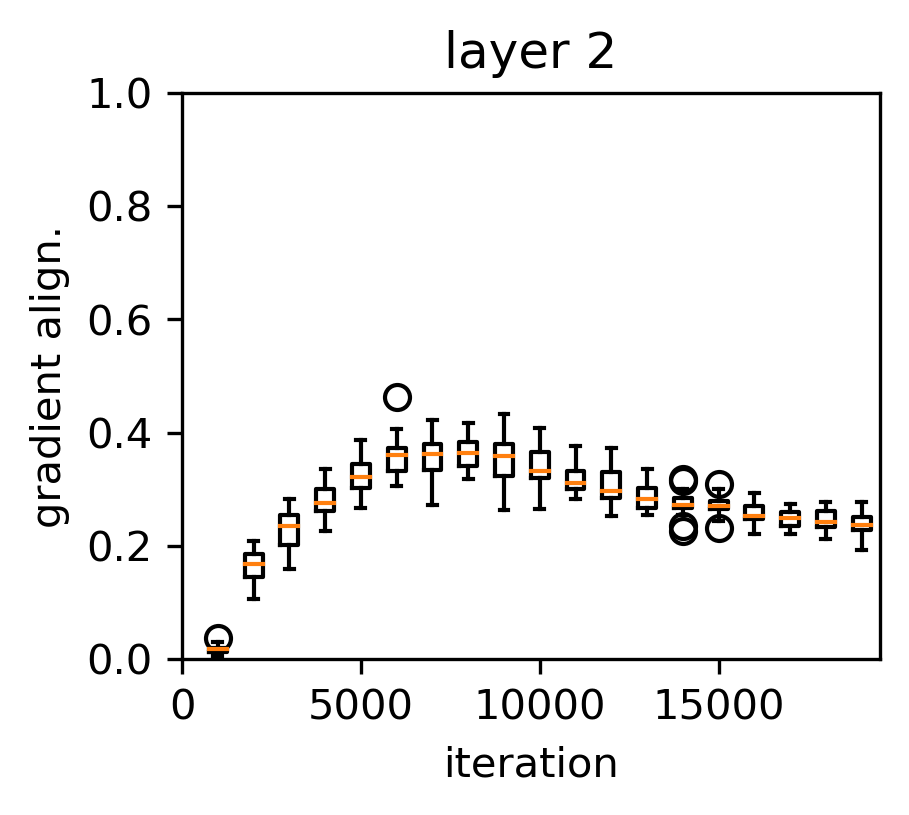}
    \includegraphics[width=0.3\linewidth]{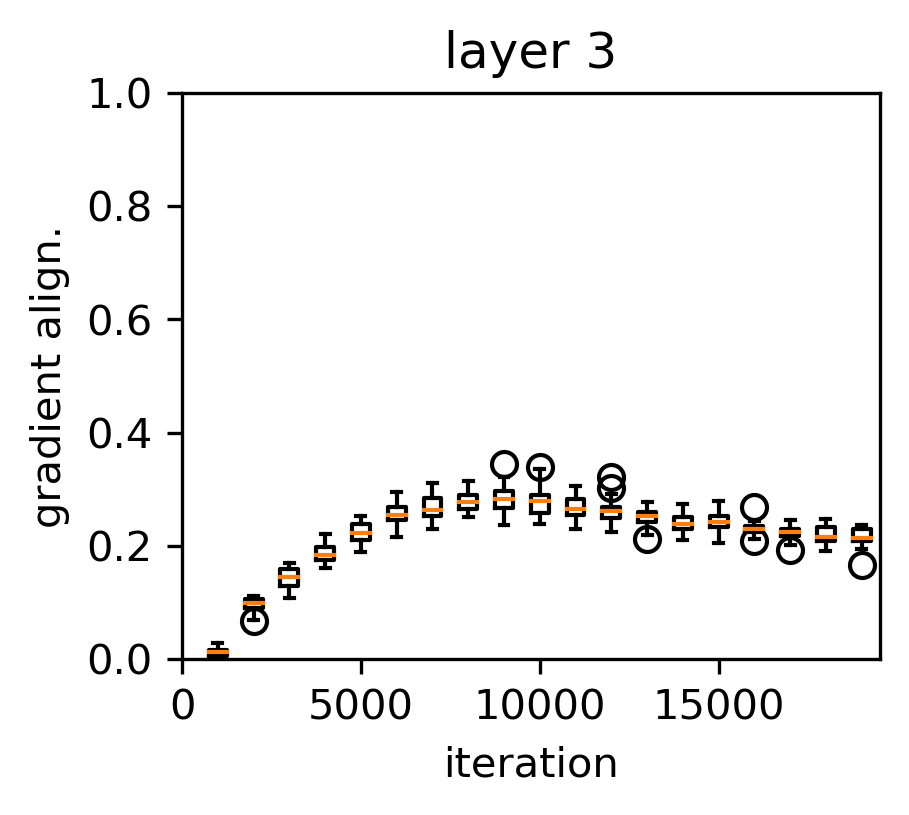}

    \includegraphics[width=0.3\linewidth]{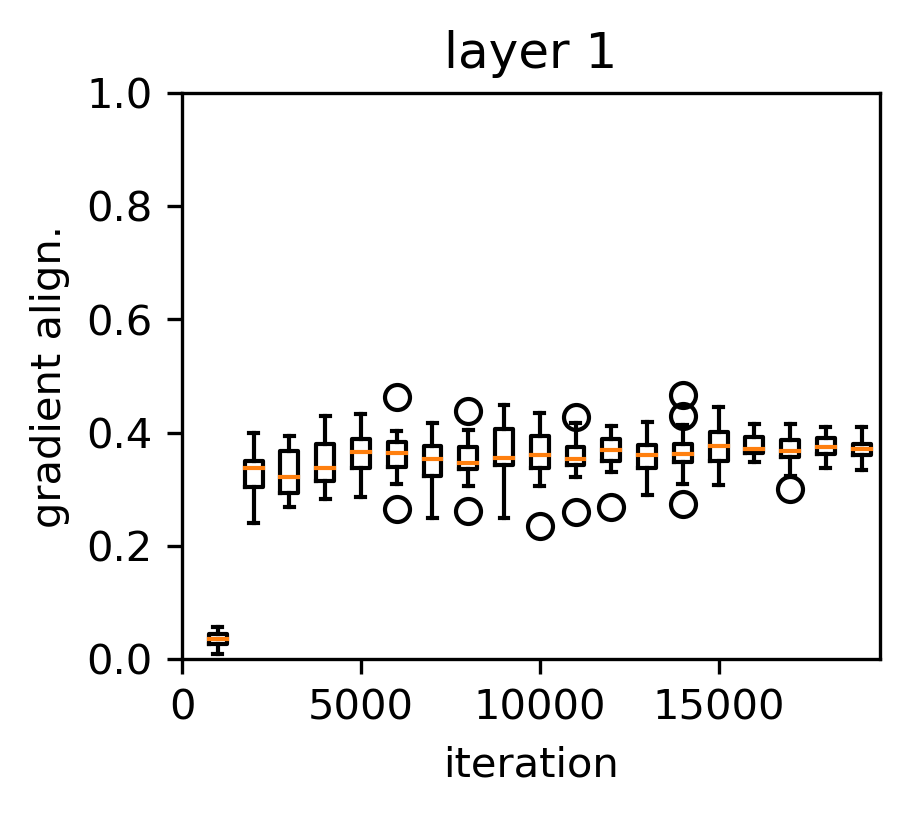}
    \includegraphics[width=0.3\linewidth]{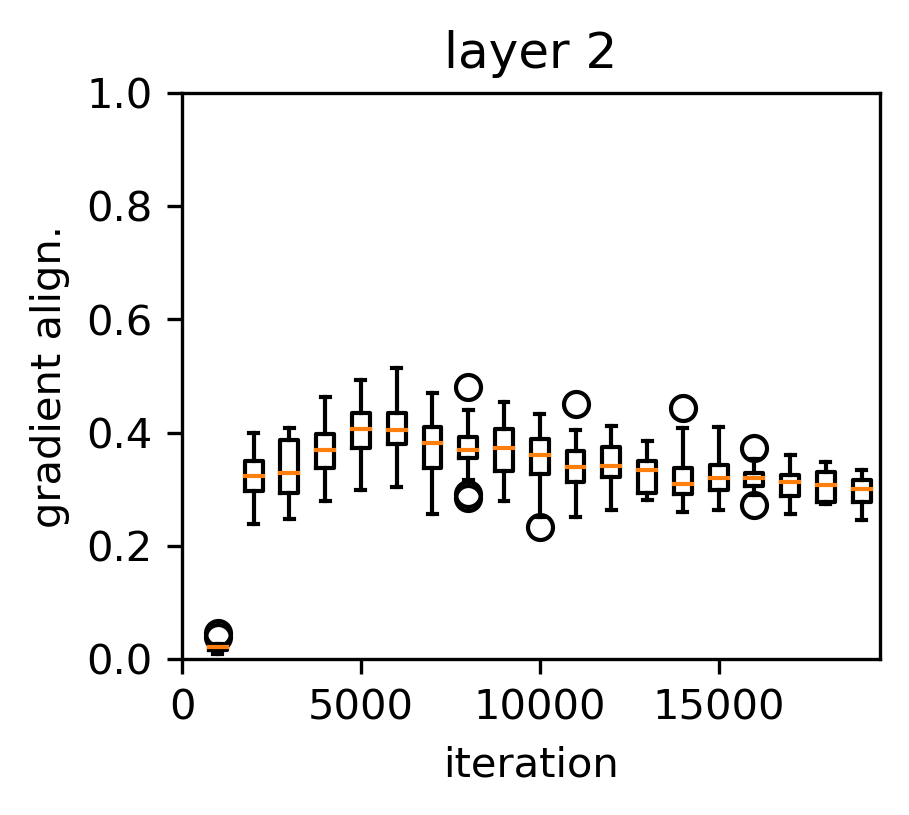}
    \includegraphics[width=0.3\linewidth]{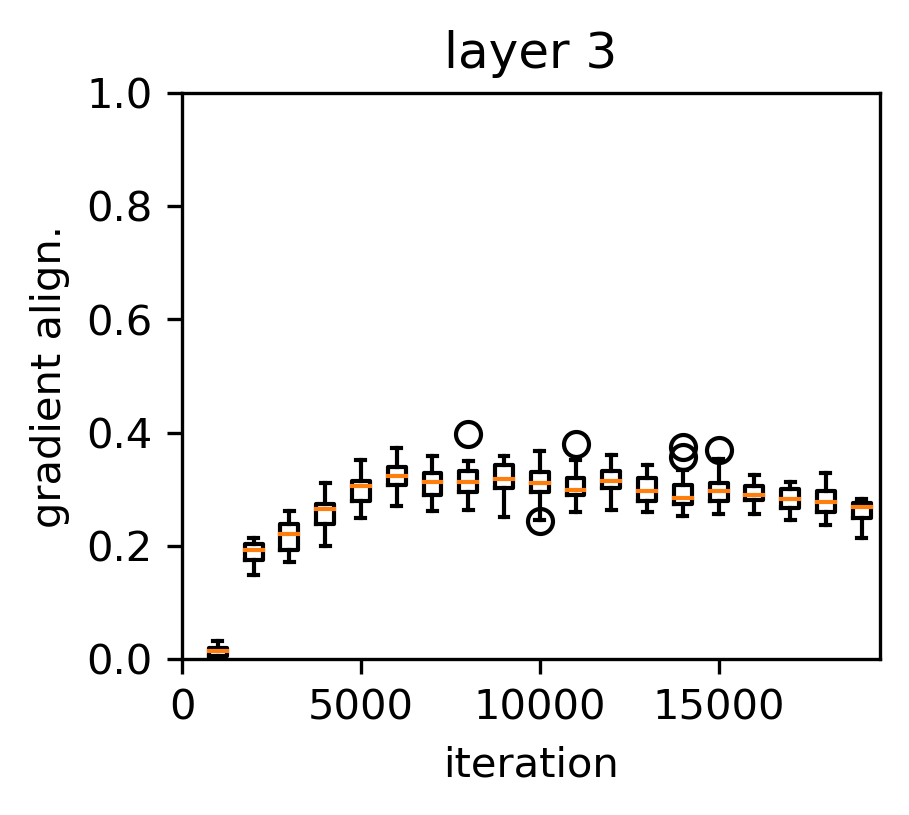}
    \caption{Gradient alignment for different layers of the step-ReLU network. Upper to lower: $Q=1,4,7$. We see that the larger the $Q$ is, the better the alignment with the ReLU becomes.}
    \label{fig:step relu gradient alignment}
\end{figure}
\subsection{Gradient Alignment}
See Figure~\ref{fig:step relu gradient alignment}.

\clearpage

\begin{figure}[t!]
    \centering
    \includegraphics[width=0.23\linewidth]{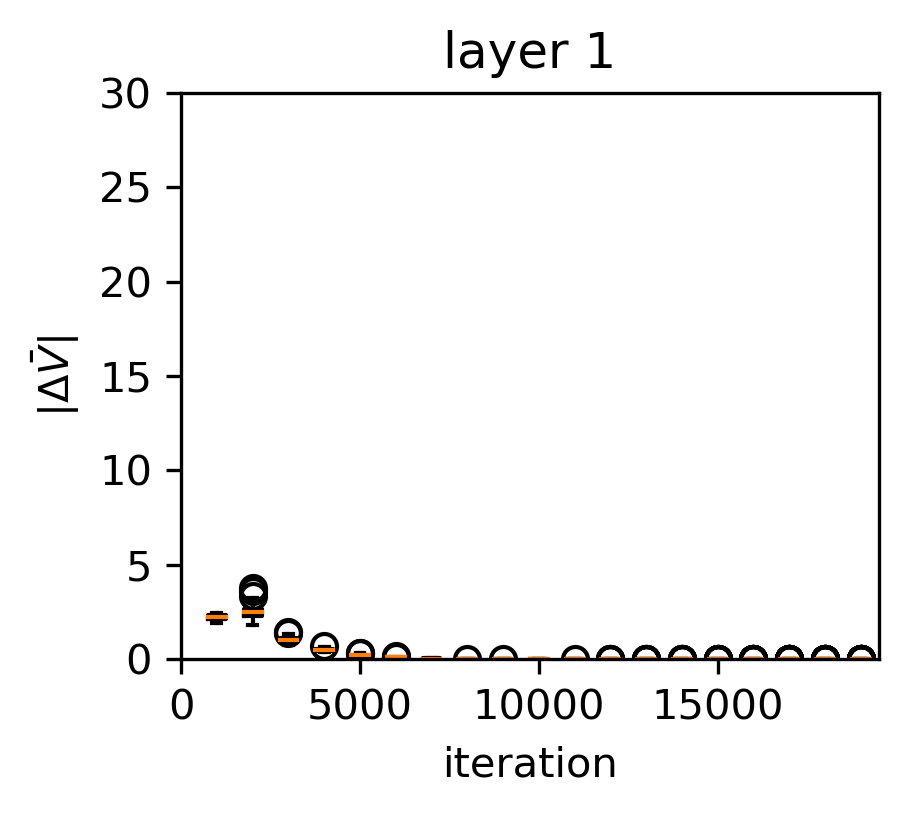}
    \includegraphics[width=0.23\linewidth]{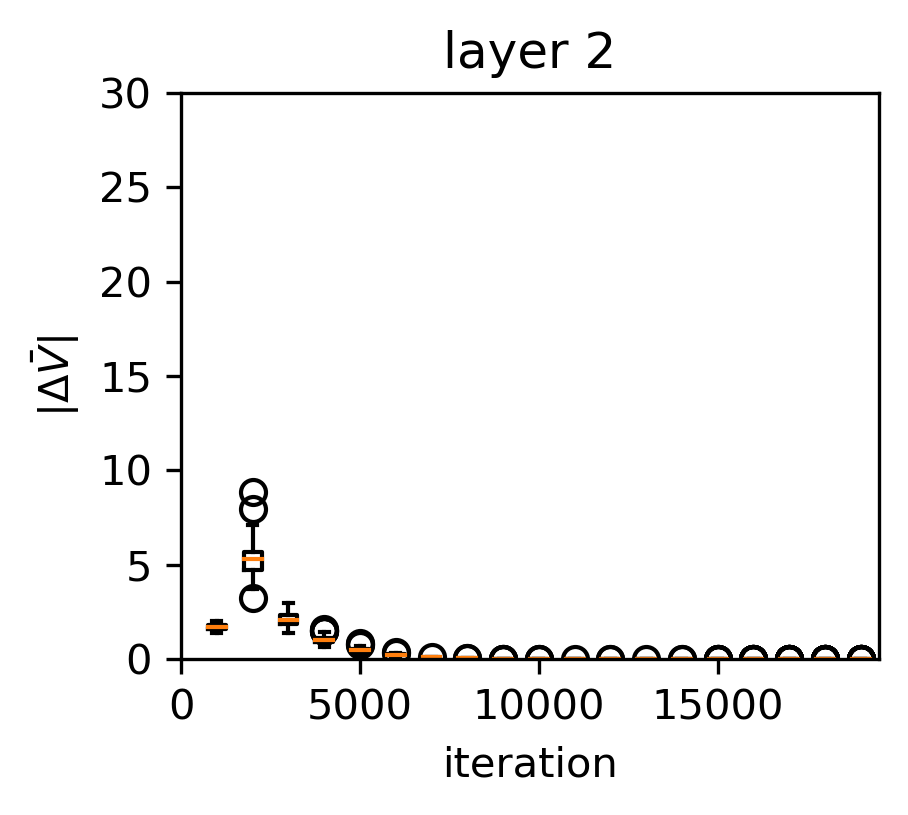}
    \includegraphics[width=0.23\linewidth]{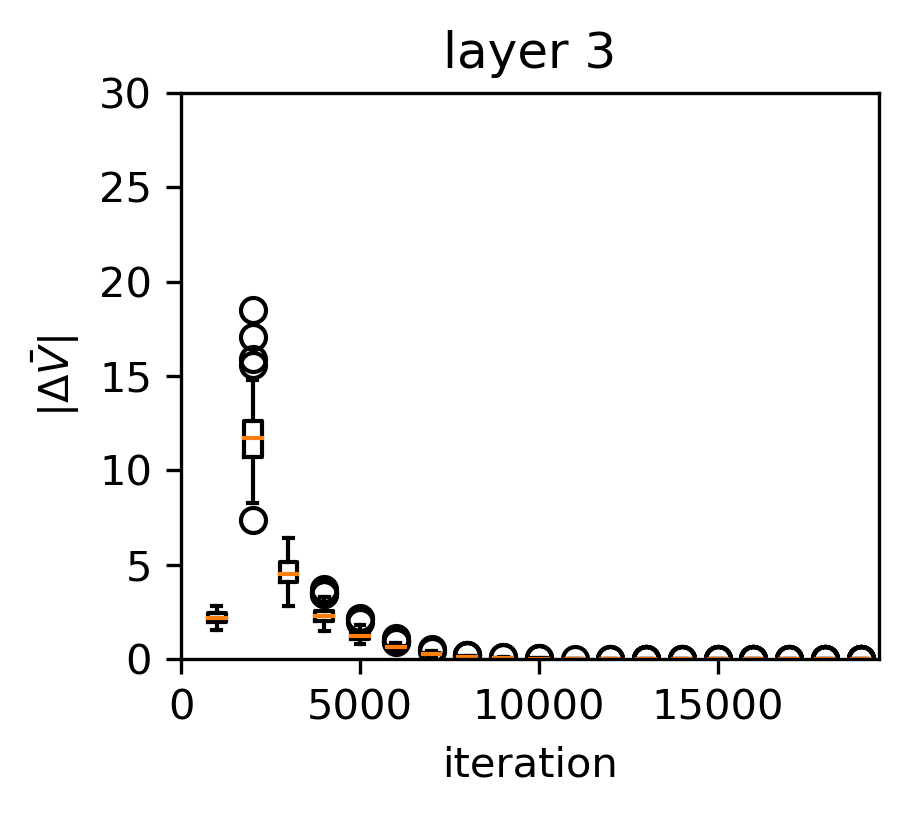}
    \includegraphics[width=0.23\linewidth]{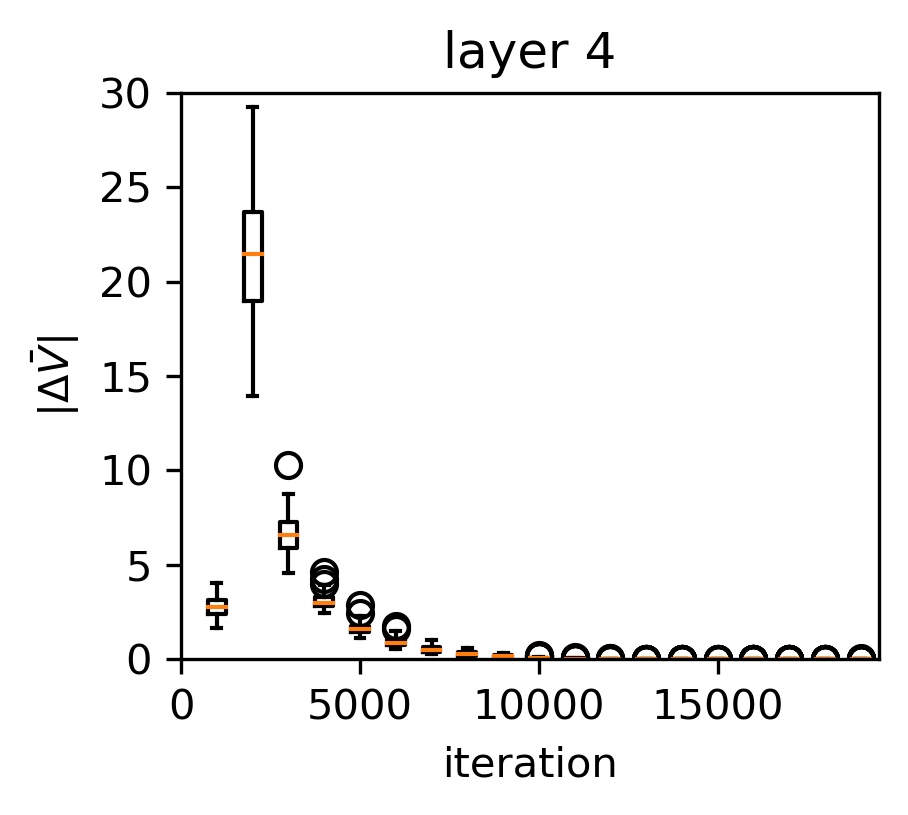}

    \includegraphics[width=0.23\linewidth]{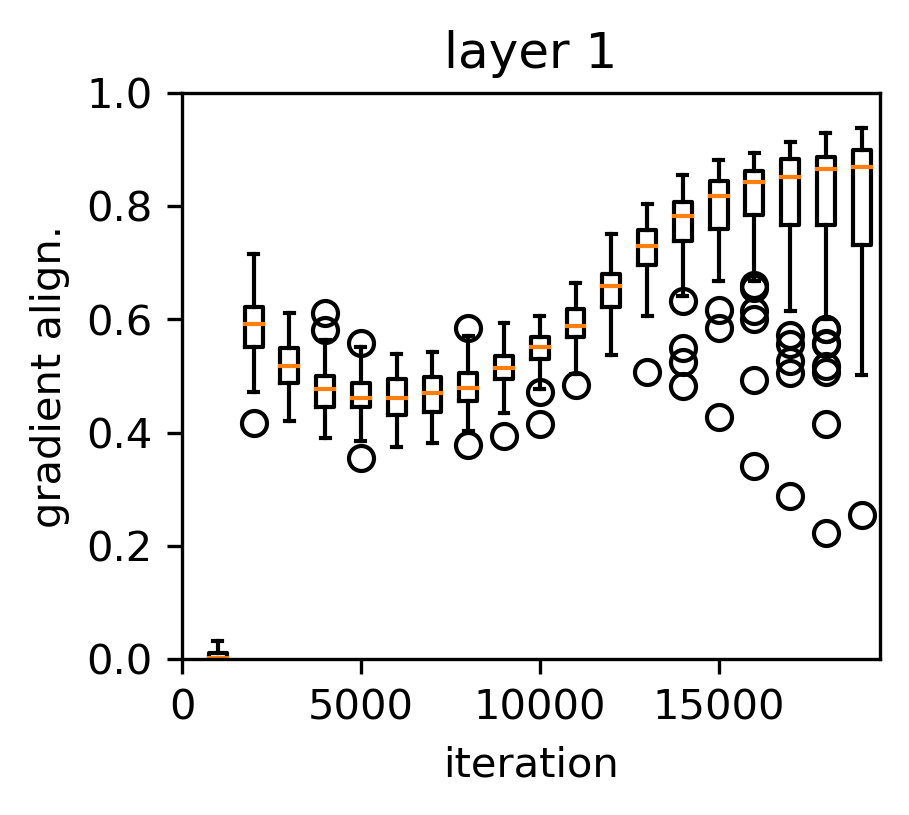}
    \includegraphics[width=0.23\linewidth]{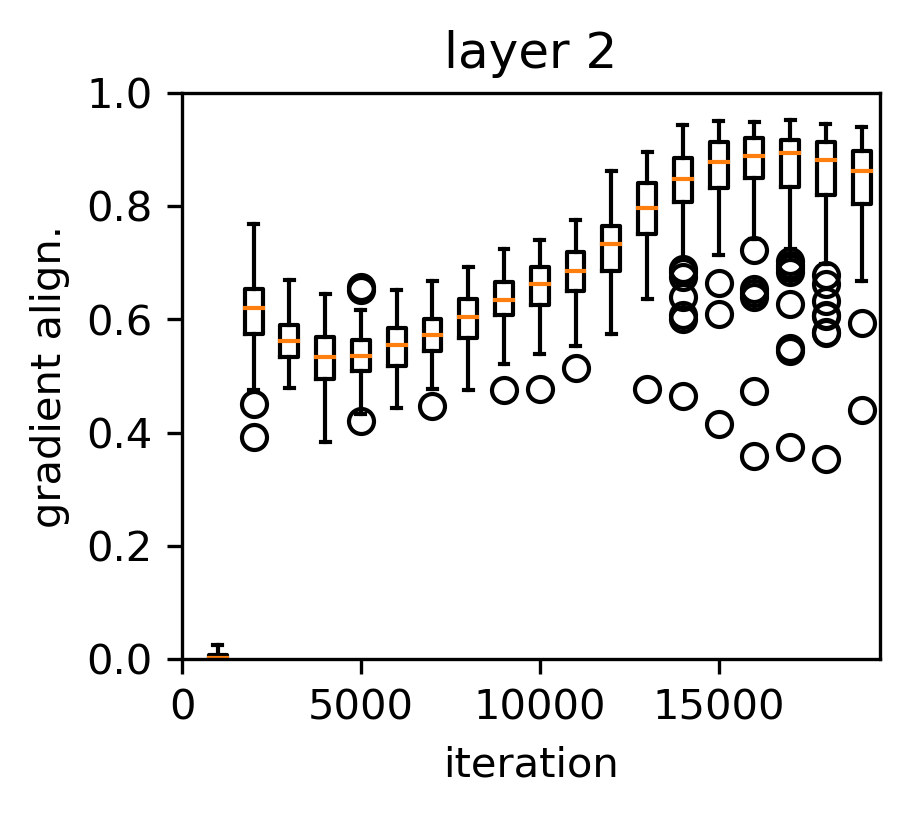}
    \includegraphics[width=0.23\linewidth]{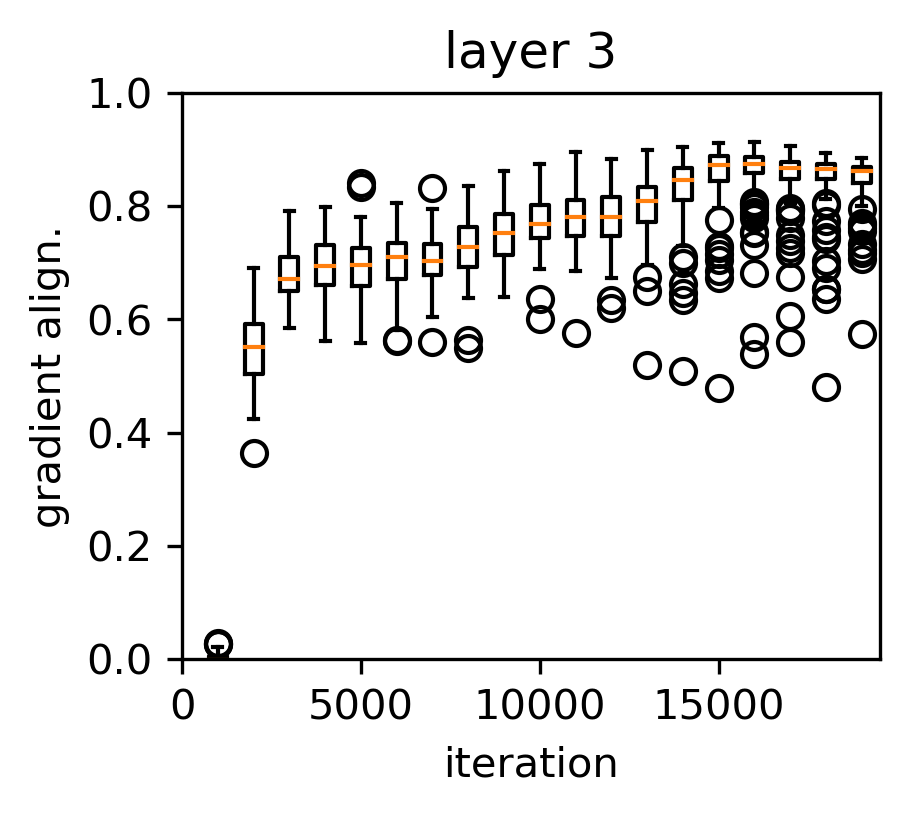}
    \includegraphics[width=0.23\linewidth]{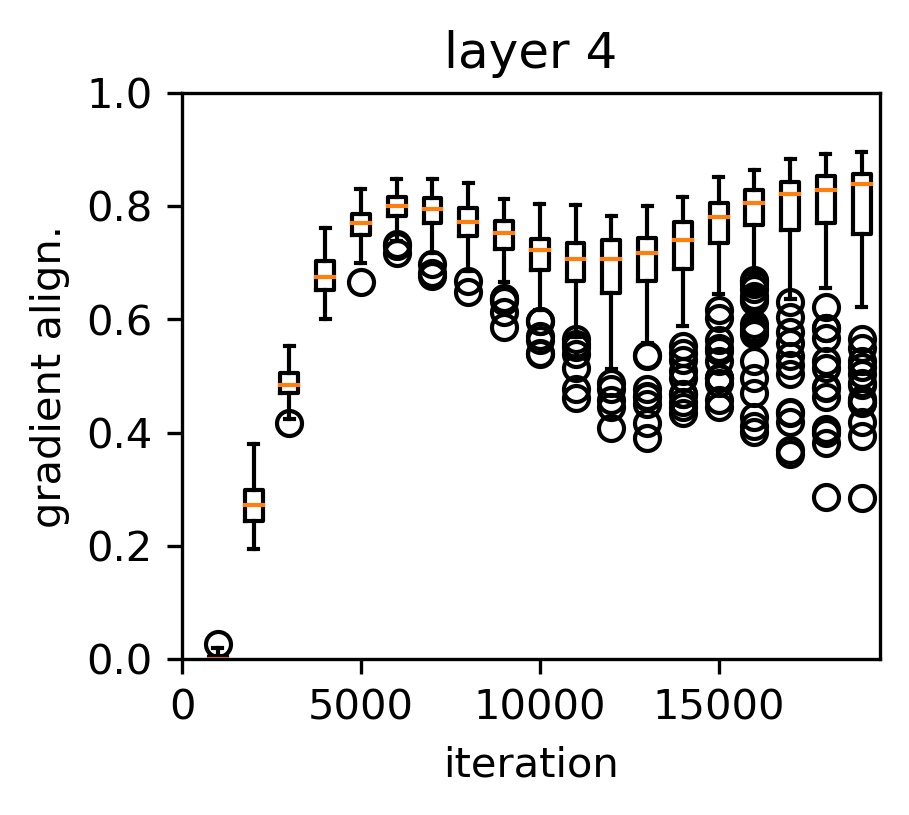}

    \caption{\small We generate a two-pathway architecture with four layers of neurons. Each neuron is treated as a node and is randomly connected to 70\% of the neurons in the next layer. The interconnections between the two pathways are completely random, and each edge is present with 70\% probability. We generate 100 such random neuronal graphs and evolve the network on the CIFAR-10 dataset. Each neuron is set to have a GeLU activation function \cite{ramachandran2017searching}, which does not have a priori dynamical consistency. \textbf{Upper}: As the training happens, the $\bar{V}$ matrix approaches stationarity across all layers and all connectivity graphs. \textbf{Lower}: The alignment of the HSP update to the gradient increases from $0$ to a level close to $O(1)$ during training for all connectivity graphs. }
    \label{fig:V stationarity and alignment}
\end{figure}

\begin{figure}[t!]
    \centering

    \centering
    \includegraphics[width=0.23\linewidth]{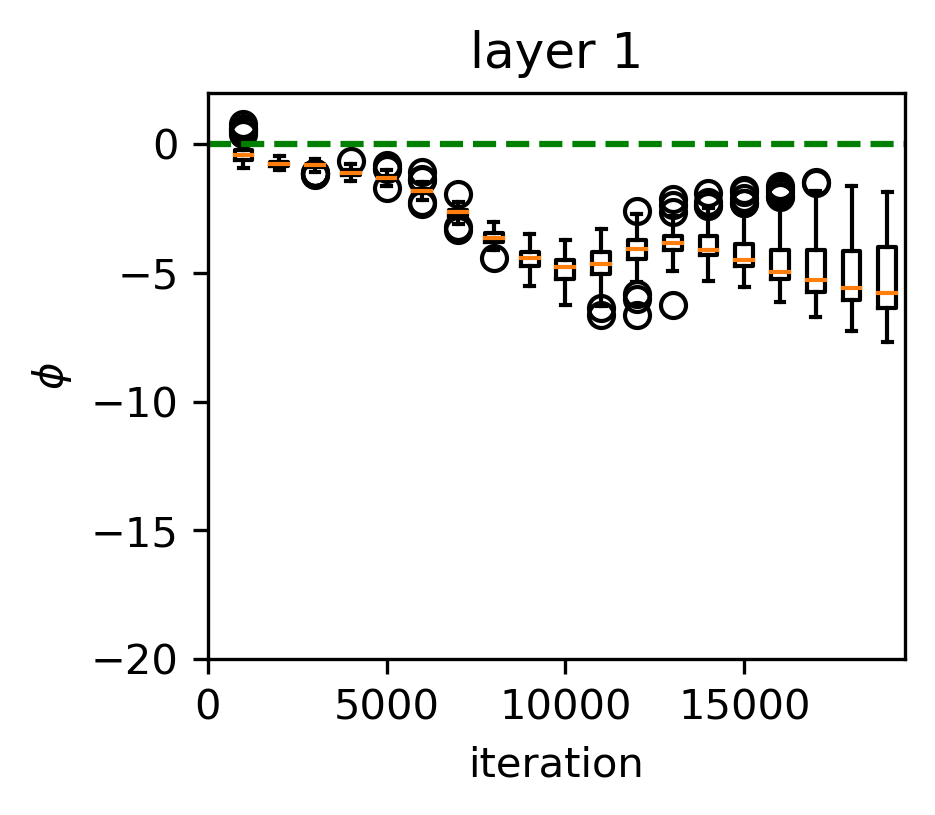}
    \includegraphics[width=0.23\linewidth]{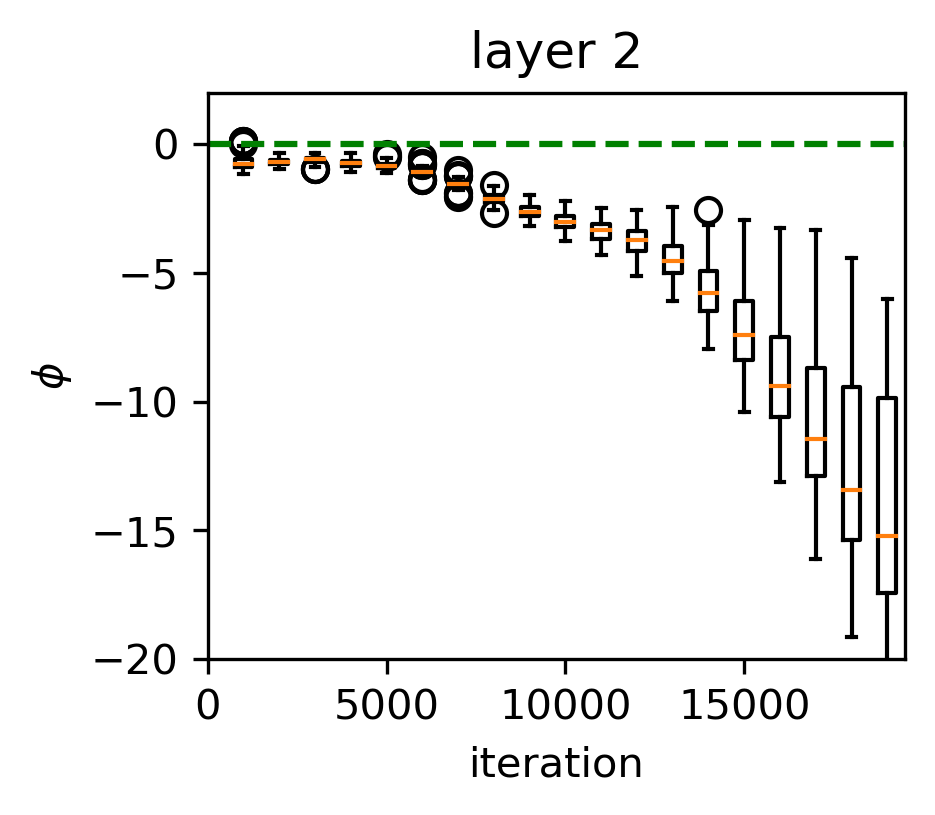}
    \includegraphics[width=0.23\linewidth]{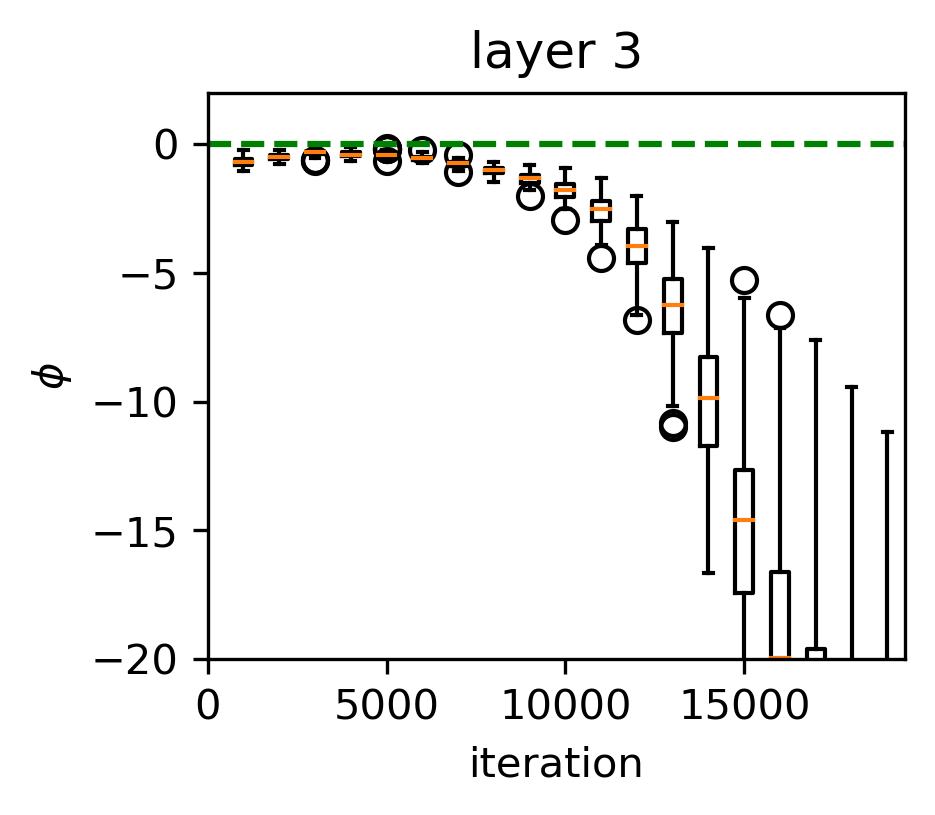}
    \includegraphics[width=0.23\linewidth]{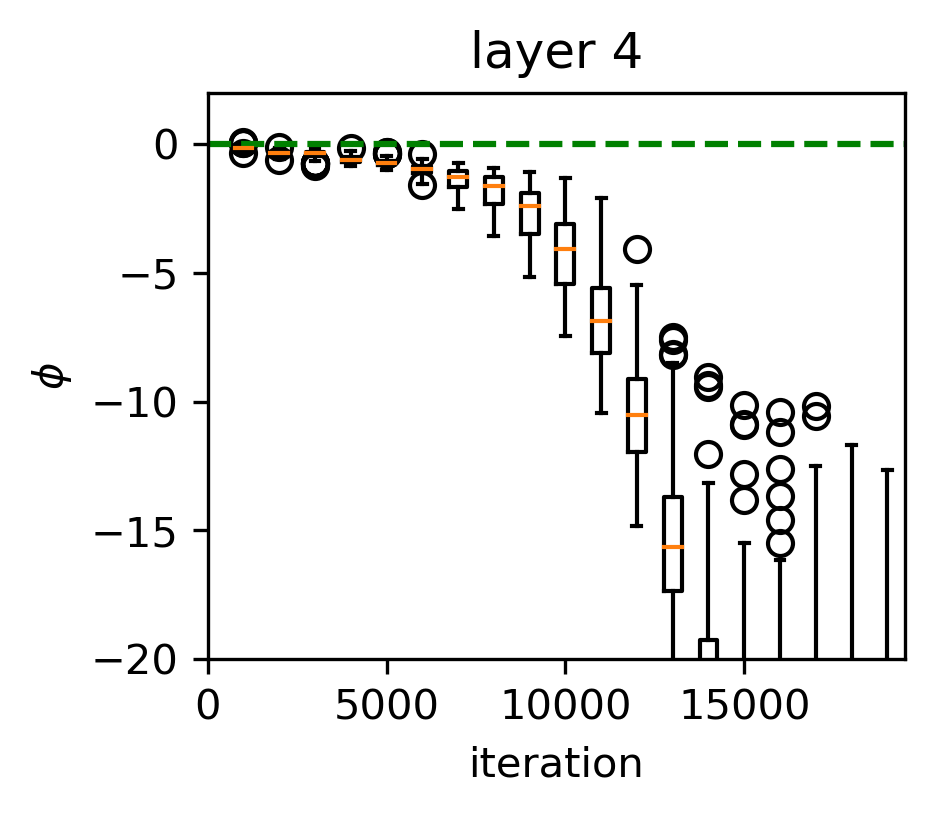}
    \caption{\small At the beginning of training, the layers are not mutually consistent but emerge to become consistent for all architectures. Some layers also start with a positive consistency score but become negative after a few hundred updates.}
    \label{fig:hebbian overlap and consistency}
\end{figure}

\subsection{Random Microscopic Connectivities}
See Figures~\ref{fig:V stationarity and alignment} and \ref{fig:hebbian overlap and consistency}.

\clearpage

\subsection{Hebbian Dynamics}

See Figure~\ref{fig:hebbian overlap zoom in}.

\begin{figure}[t!]
    \centering

    \includegraphics[width=0.23\linewidth]{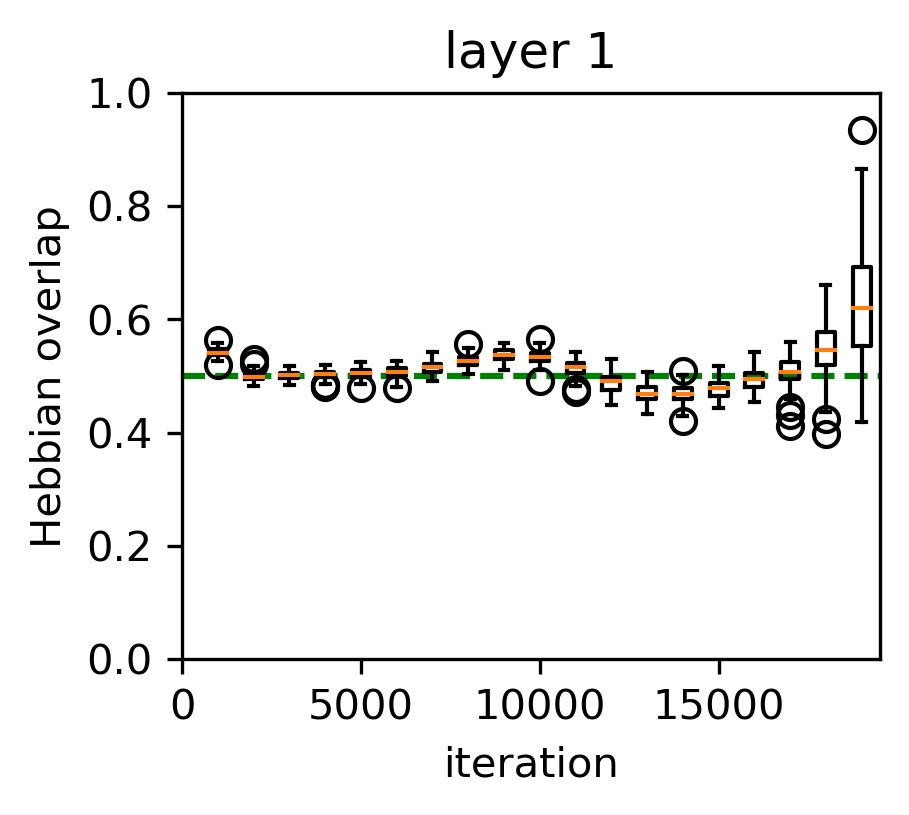}
    \includegraphics[width=0.23\linewidth]{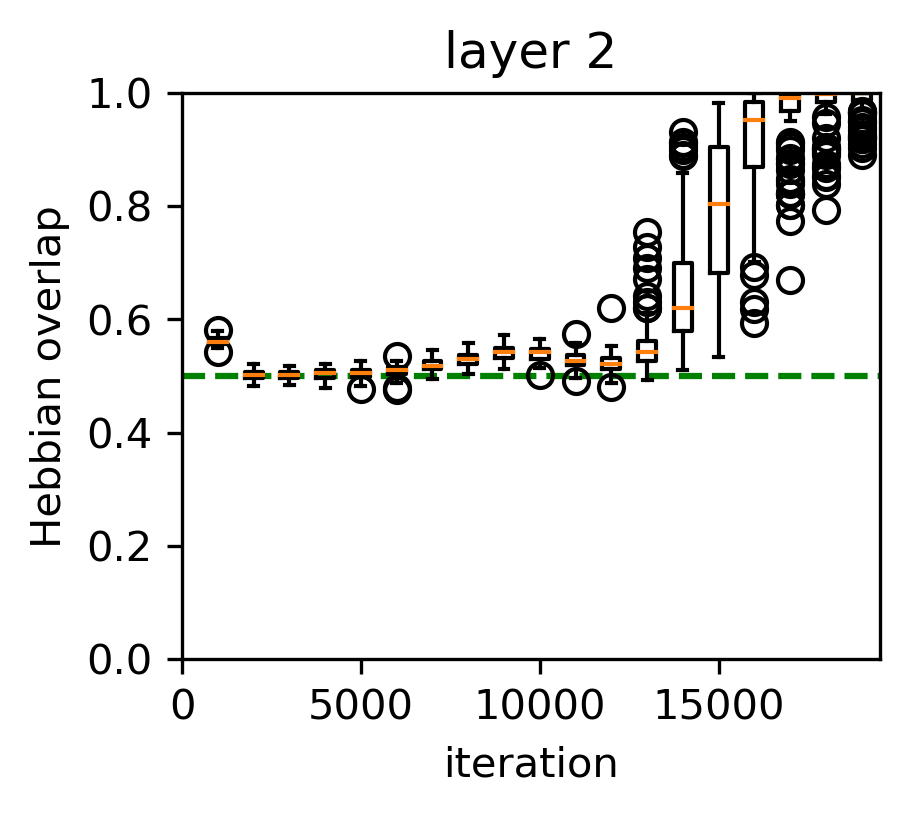}
    \includegraphics[width=0.23\linewidth]{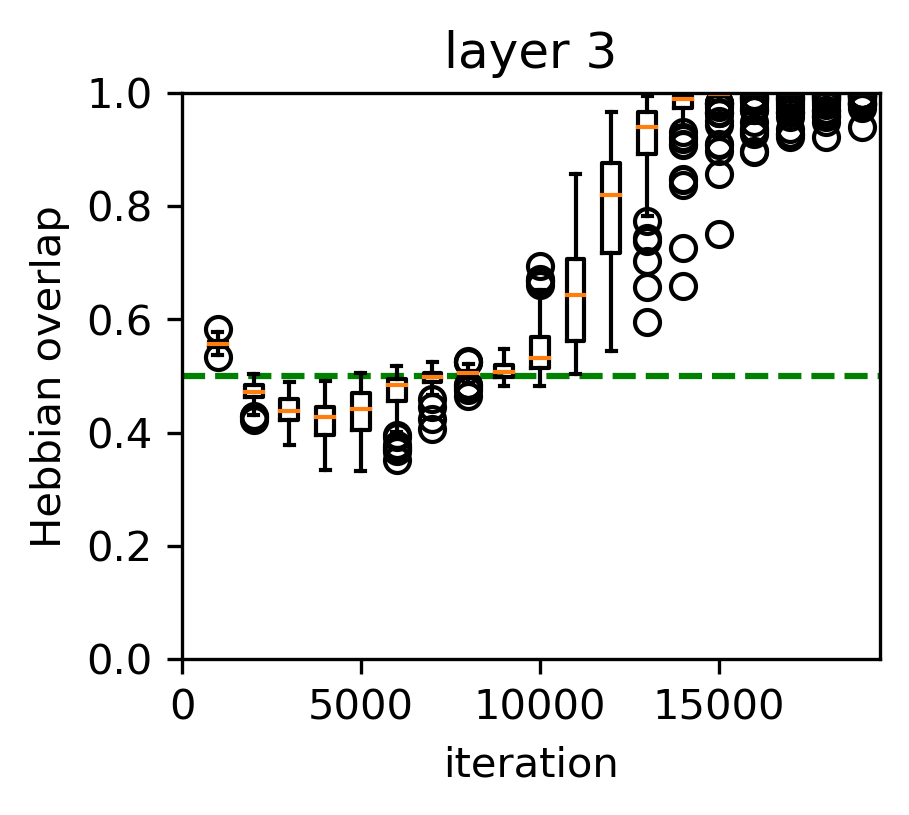}
    \includegraphics[width=0.23\linewidth]{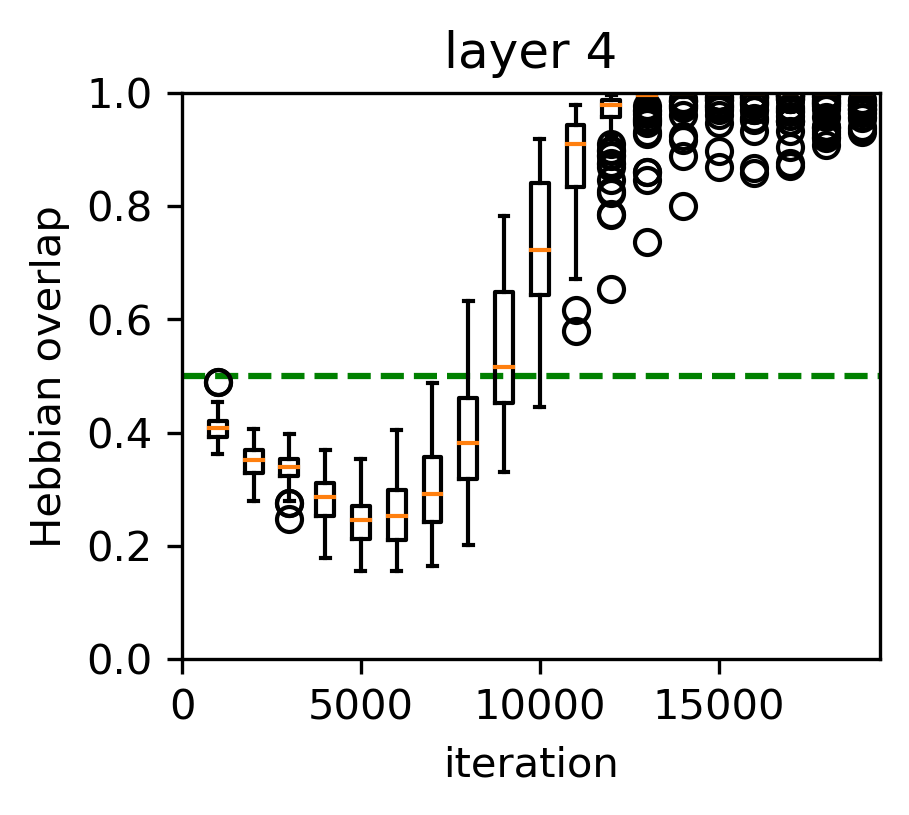}
    
    \includegraphics[width=0.23\linewidth]{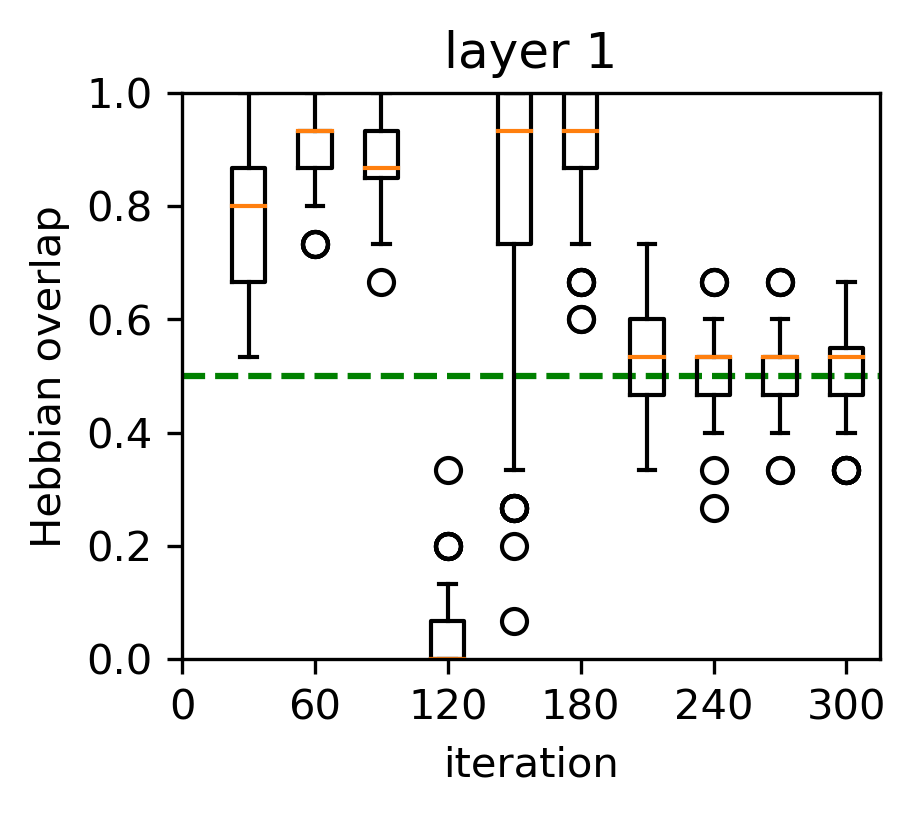}
    \includegraphics[width=0.23\linewidth]{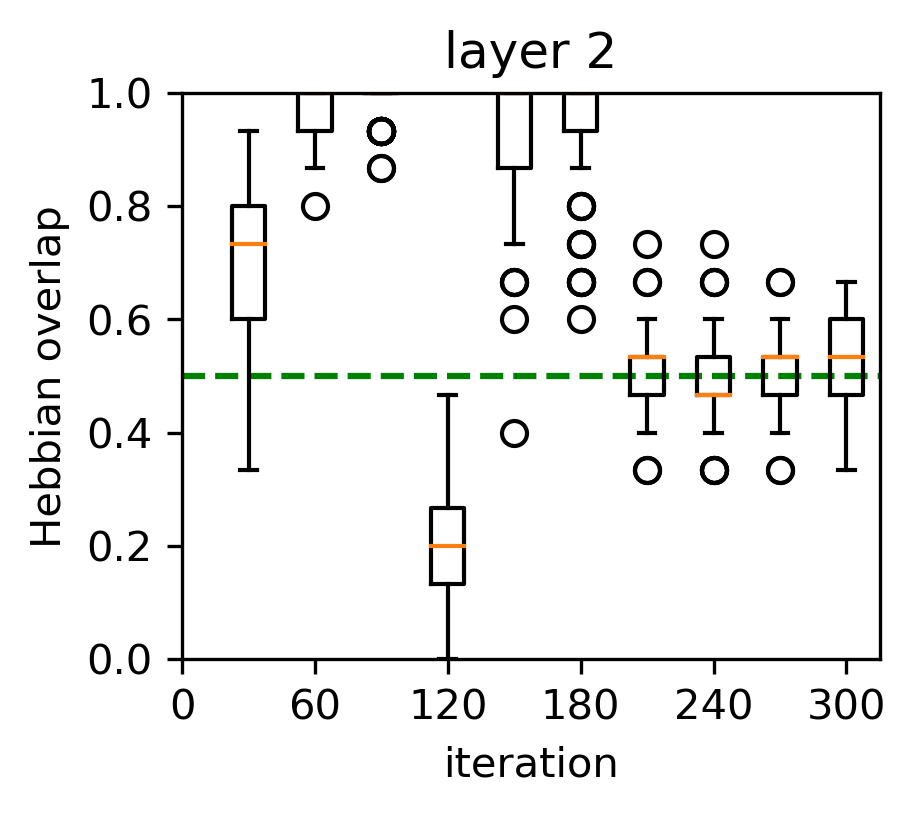}
    \includegraphics[width=0.23\linewidth]{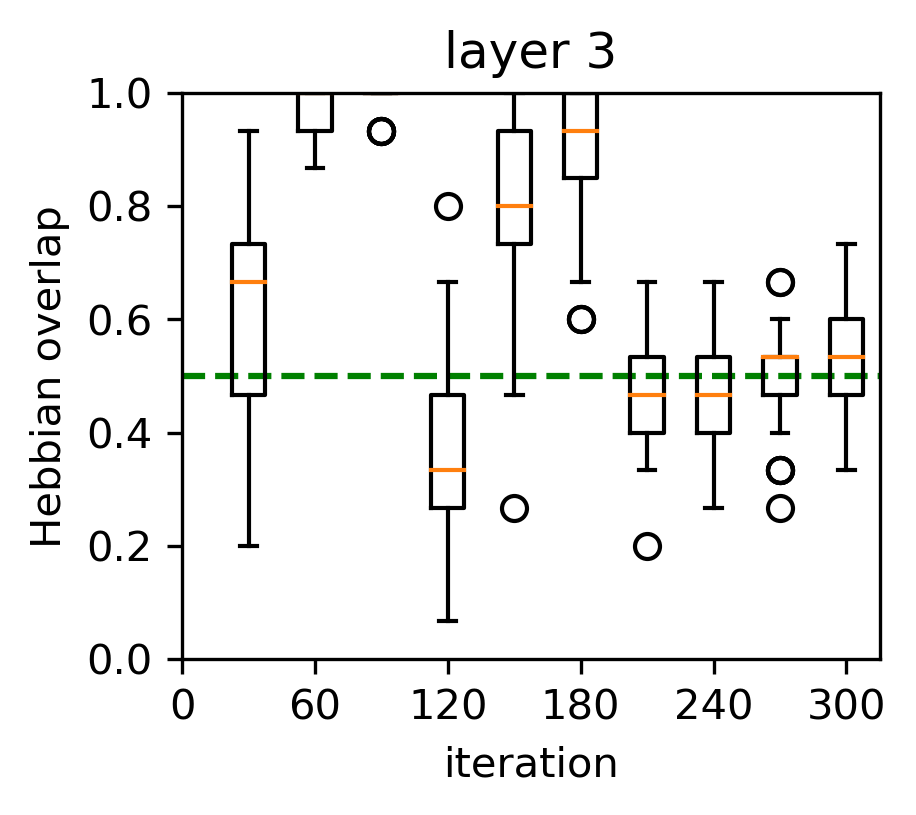}
    \includegraphics[width=0.23\linewidth]{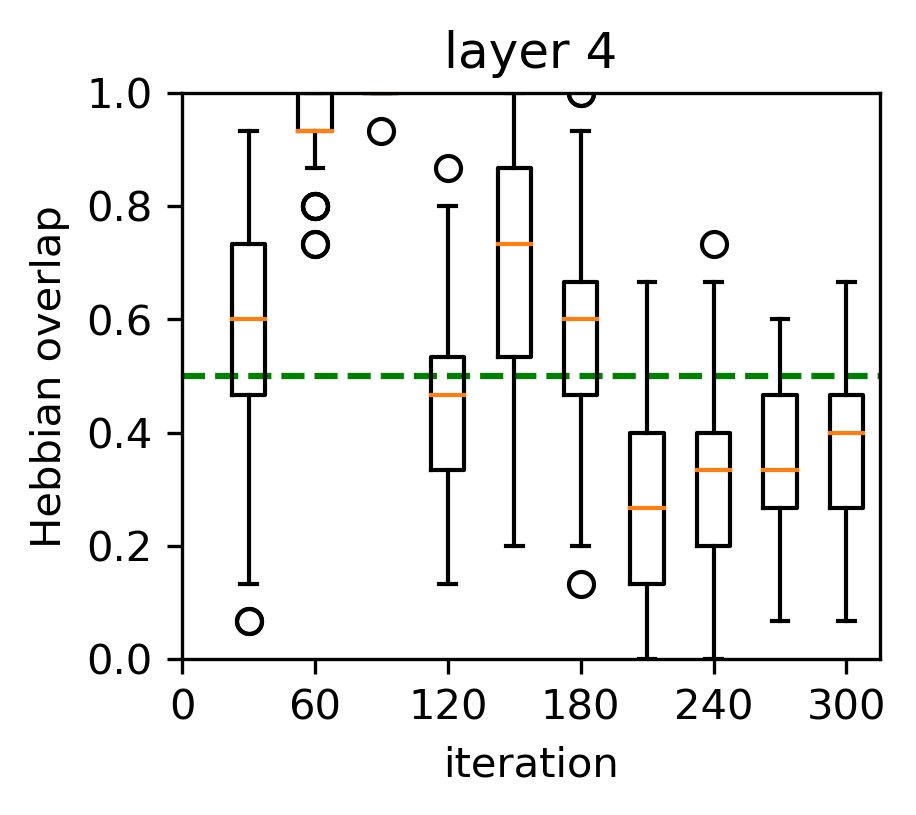}

    \caption{\small Experimental setting same as Figure~\ref{fig:V stationarity and alignment}. \textbf{Upper}: The proportion of the updates positively correlated with the Hebbian rule. All layers evolve to significantly overlap with the Hebbian rule, with later layers having better overlap. \textbf{Lower}: This is a zoom-in of the first 300 iterations.We see that all layers feature an initial transient anti-Hebbian phase, which transitions to Hebbian quickly after for the layer layers. For initial layers, the transition is much slower.}
    \label{fig:hebbian overlap zoom in}
\end{figure}

\clearpage

\clearpage
\subsection{High Level Structures} 

See Figure~\ref{fig:cifar10 btf sign}-\ref{fig:cifar10 random}.

\begin{figure}[t!]
    \centering
    \includegraphics[width=0.22\linewidth]{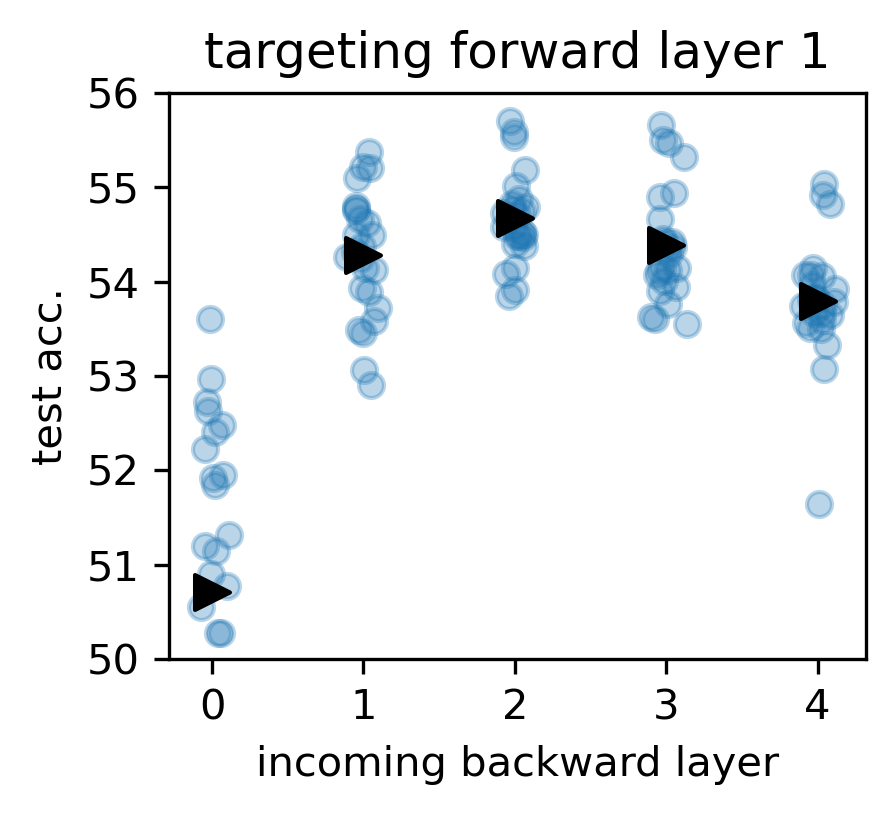}
    \includegraphics[width=0.22\linewidth]{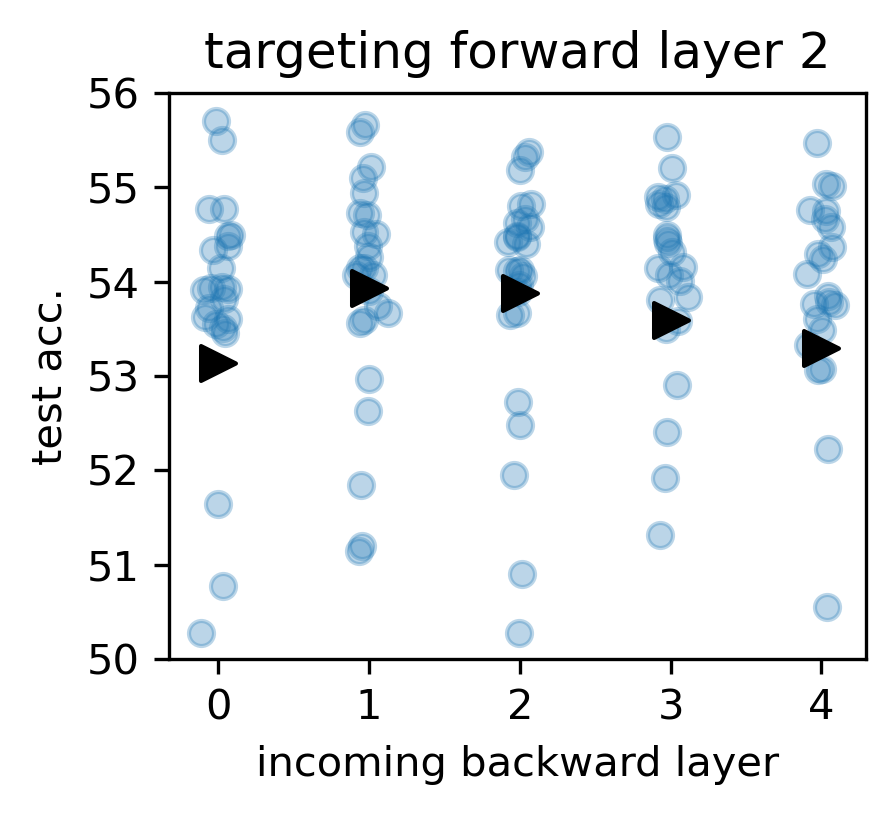}
    \includegraphics[width=0.22\linewidth]{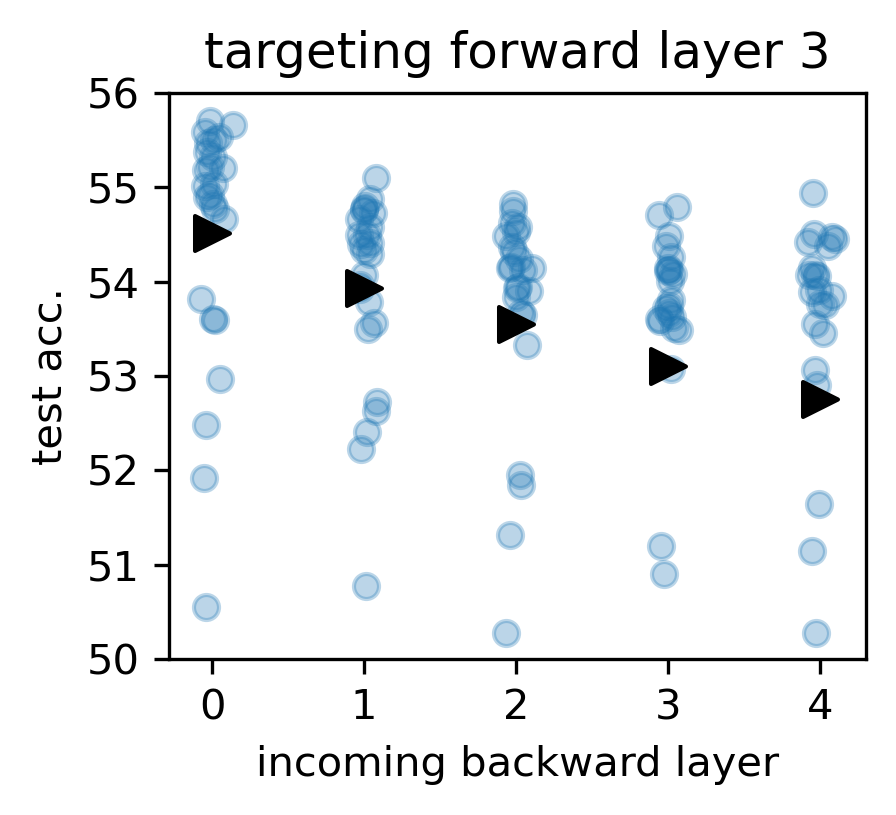}

    \caption{An exhaustive search for \textbf{plastic} backward-to-forward connections between a five-layer forward and backward pathway. The figure shows the performance of the model conditioning on different incoming backward layers to the layer $1$ (\textbf{left}), $2$ (\textbf{middle}) and $3$ (\textbf{right}) of the forward network pathway. 
    }
    \label{fig:cifar10 btf relu}
\end{figure}

\begin{figure}[t!]
    \centering
    \includegraphics[width=0.22\linewidth]{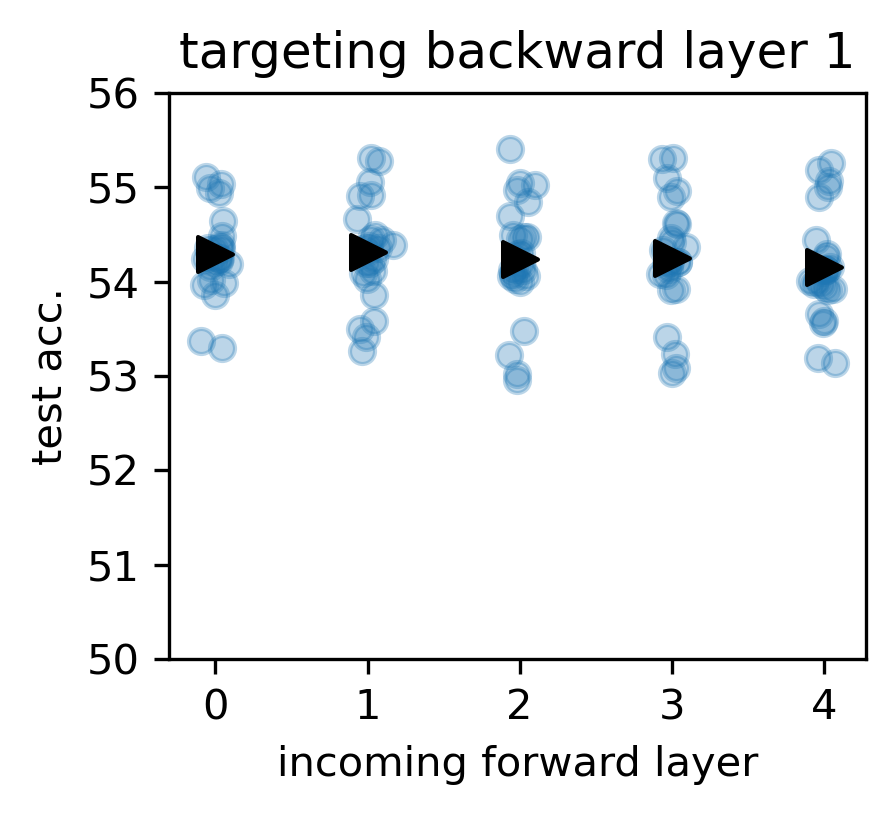}
    \includegraphics[width=0.22\linewidth]{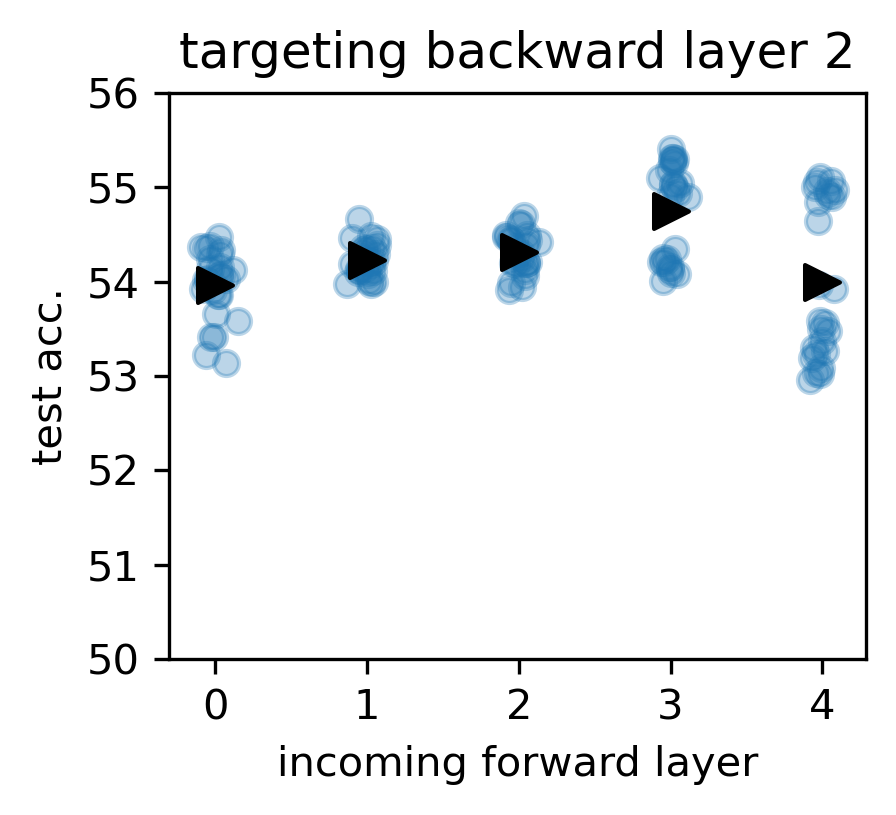}
    \includegraphics[width=0.22\linewidth]{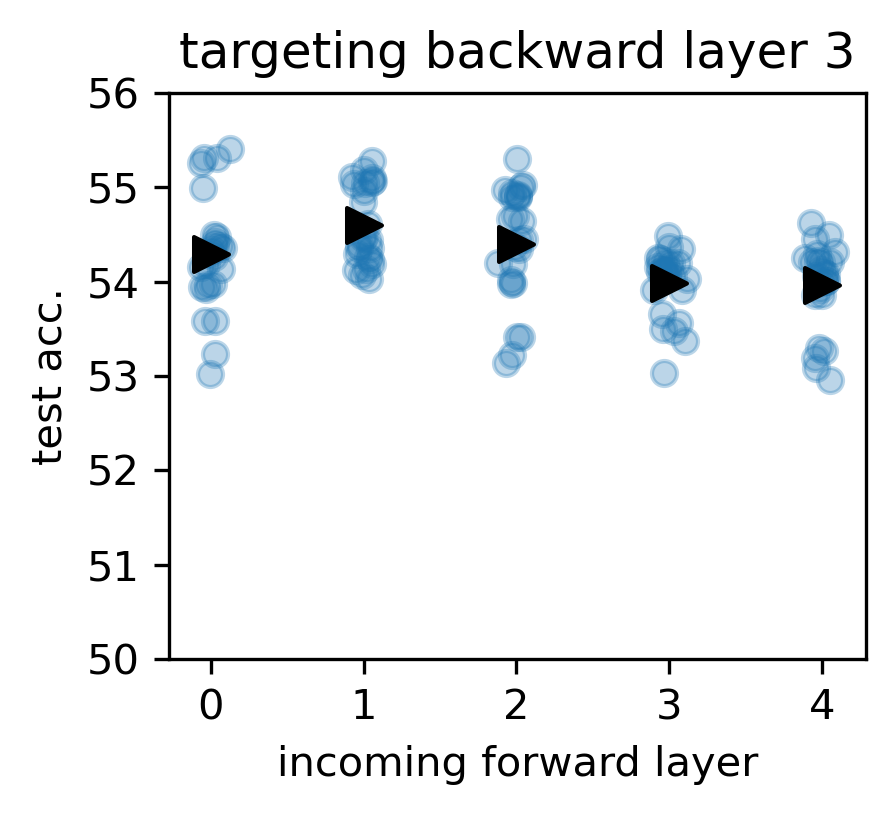}
    
    \caption{An exhaustive search for \textbf{plastic} forward-to-backward connections. The figure shows the performance of the model conditioning on different incoming forward layers to the backward network pathway. 
    }
    \label{fig:cifar10 ftb}
\end{figure}

\begin{figure}[t!]
    \centering
    \includegraphics[width=0.22\linewidth]{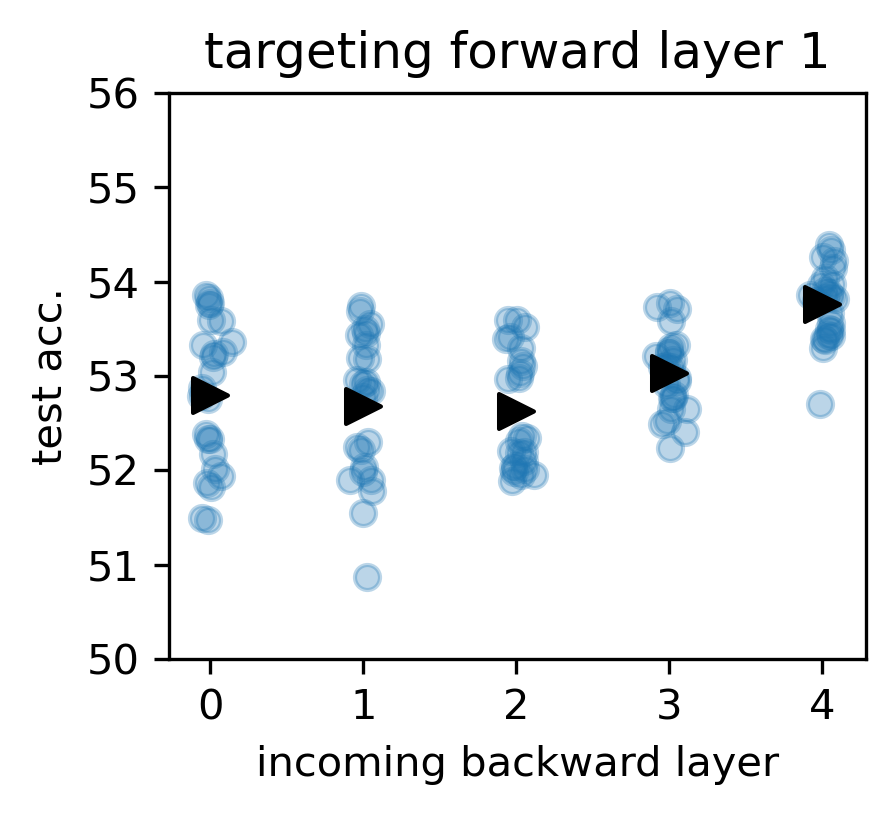}
    \includegraphics[width=0.22\linewidth]{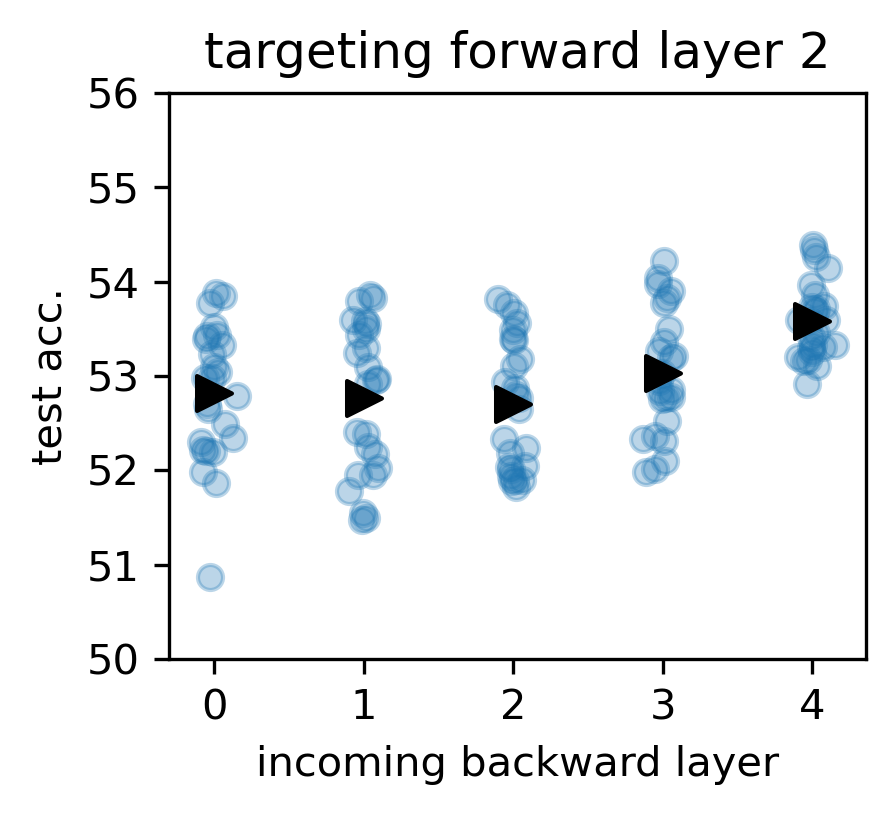}
    \includegraphics[width=0.22\linewidth]{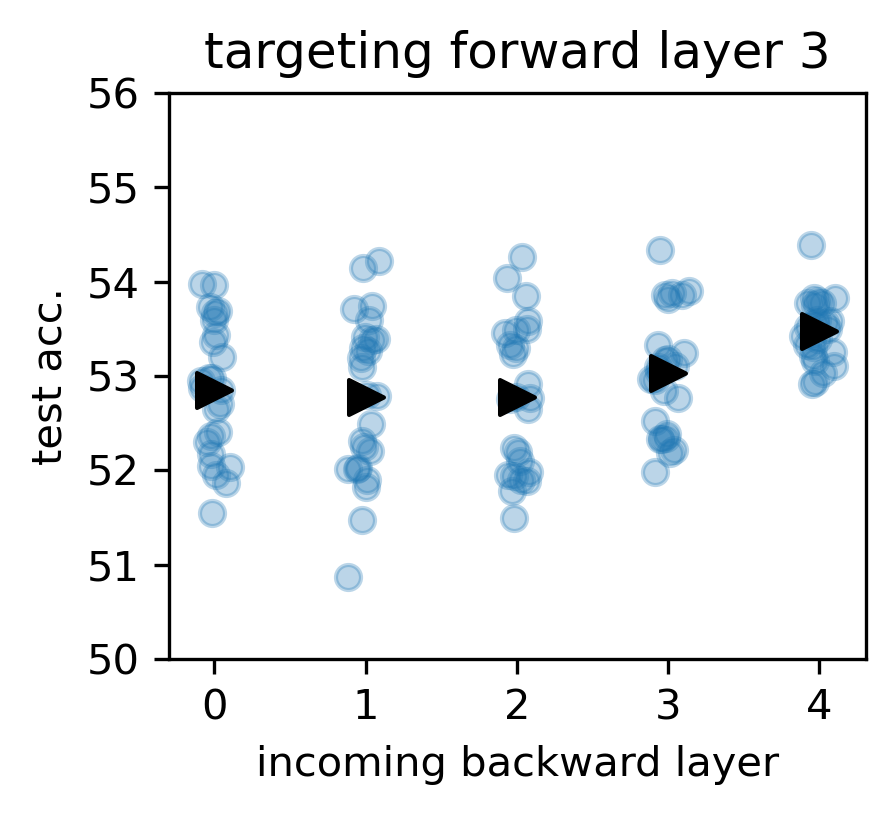}
    
    \caption{An exhaustive search for \textbf{nonplastic} backward-to-forward connections. The figure shows the performance of the model conditioning on different incoming backward layers to the forward network pathway. 
    }
    \label{fig:cifar10 random}
\end{figure}

\begin{figure}[t!]
    \centering
    \includegraphics[width=0.22\linewidth]{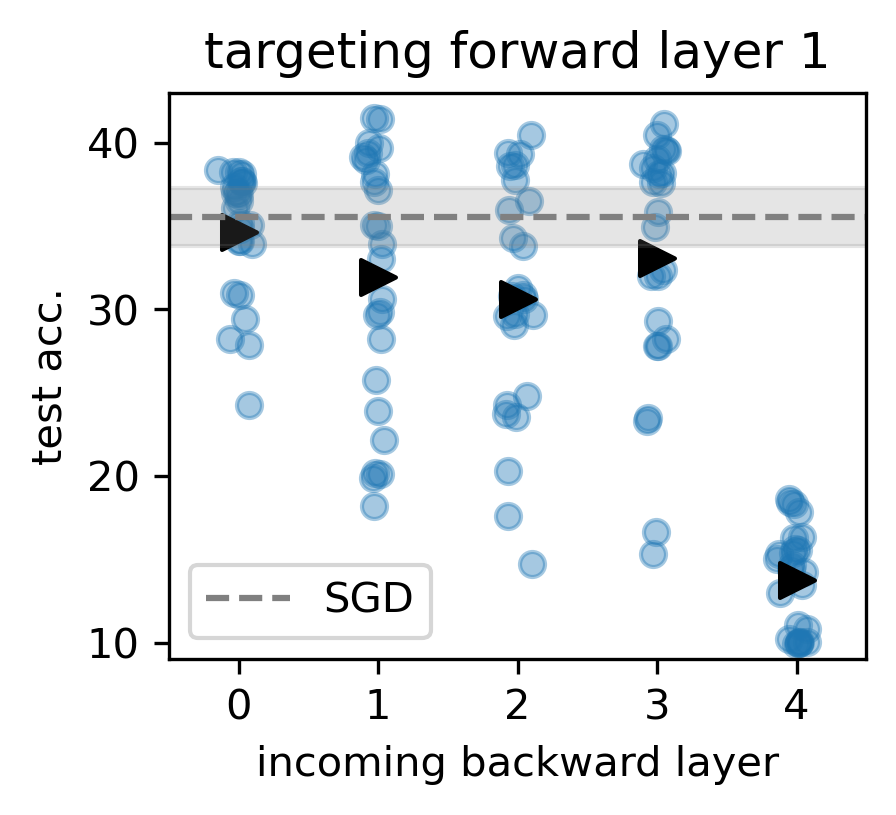}
    \includegraphics[width=0.22\linewidth]{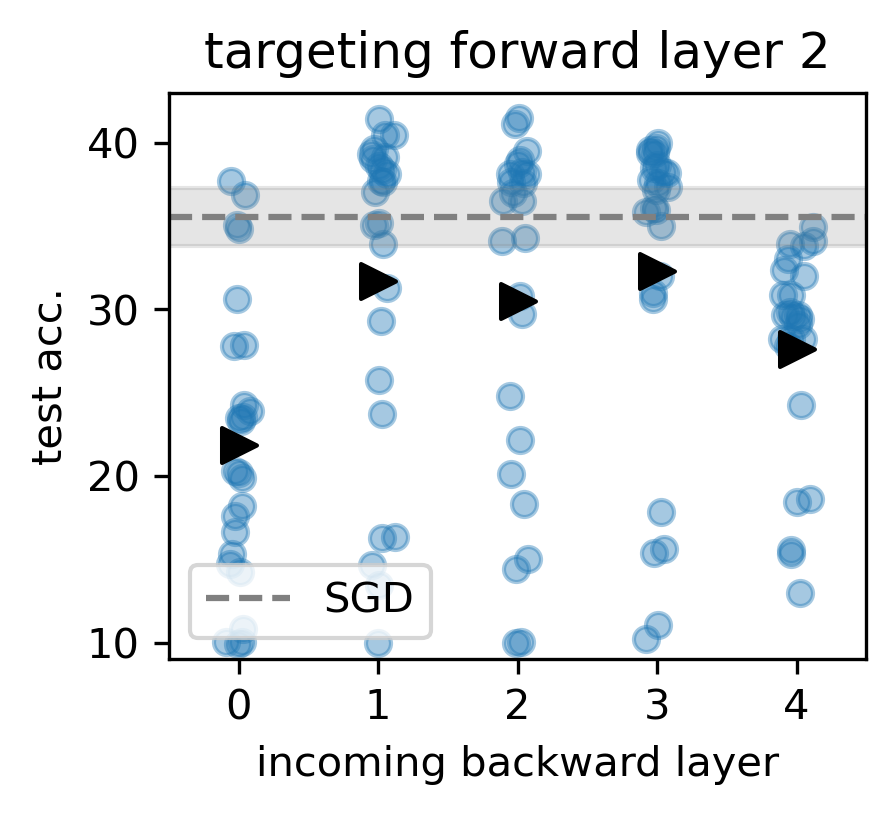}
    \includegraphics[width=0.22\linewidth]{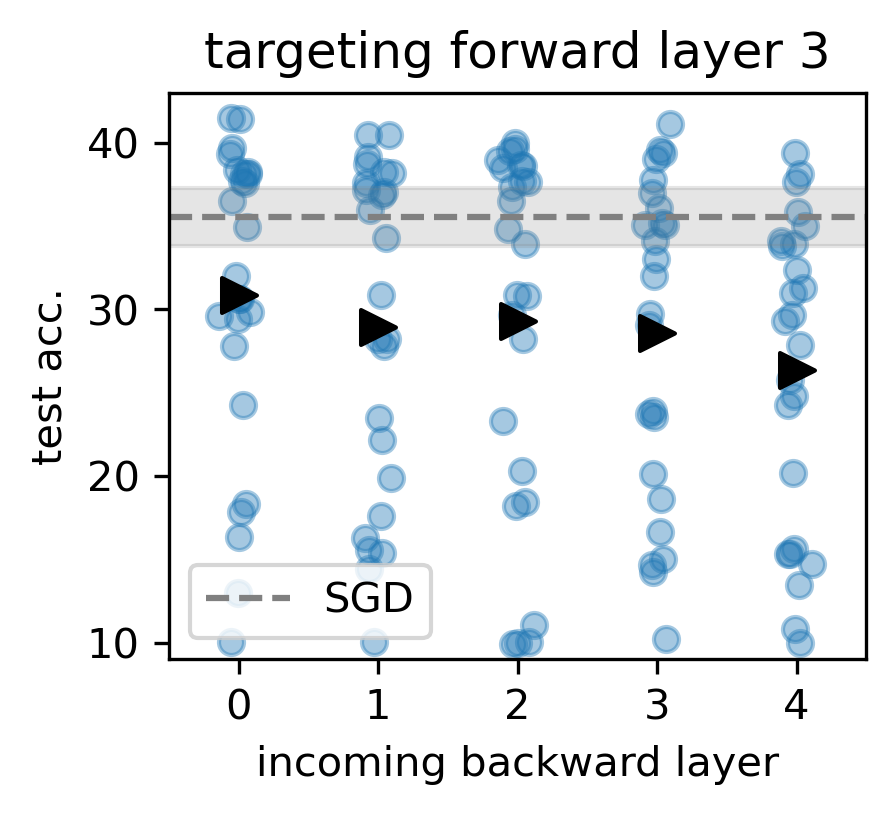}

    \caption{An exhaustive search for \textbf{plastic} backward-to-forward connections between a five-layer forward and backward pathway with sign activation. 
    }
    \label{fig:cifar10 btf sign}
\end{figure}

\clearpage
\subsection{Best and Worst Circuits}
We show the best connectivity structures found in the high-level search experiment. See Figure~\ref{fig:best worst circuits}.

\begin{figure}
    \centering
    \includegraphics[width=0.7\linewidth]{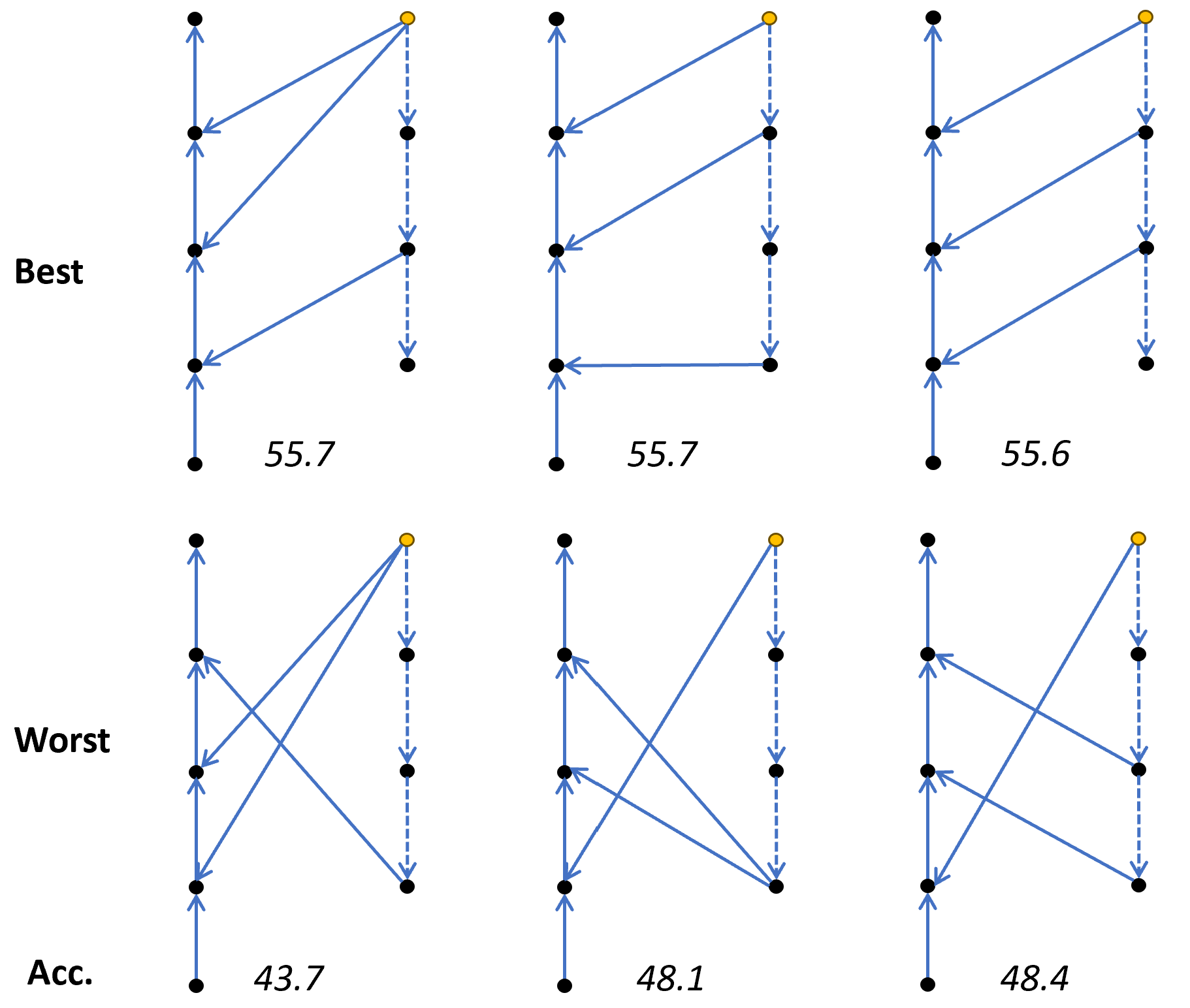}
    \caption{The three best and worst connectivity structures found in the high-level search experiment.}
    \label{fig:best worst circuits}
\end{figure}

\end{document}